\documentclass[oneside]{amsart}

\makeatletter
\g@addto@macro{\endabstract}{\@setabstract}
\newcommand{\authorfootnotes}{\renewcommand\thefootnote{\@fnsymbol\c@footnote}}%
\makeatother

\usepackage{url}

\usepackage[utf8]{inputenc}
\usepackage[T1]{fontenc}
\usepackage{amsthm}
\usepackage{mathtools}
\usepackage{amsfonts}
\usepackage{amsmath}
\usepackage{amssymb}
\usepackage{xfrac}
\usepackage{enumitem}
\usepackage{graphicx}
\usepackage{mathrsfs}
\usepackage{subcaption}
\usepackage[margin=3cm]{geometry}

\usepackage{etoolbox}

\usepackage[
    unicode=true,
    bookmarks=true,
    bookmarksopen=false,
    breaklinks=false,
    pdfborder={0 0 0},
    colorlinks=true
]{hyperref}

\usepackage{xcolor}
\definecolor{cblue}{rgb}{0.16, 0.32, 0.75}
\definecolor{cred}{rgb}{0.7, 0.11, 0.11}
\hypersetup{%
	,linkcolor=cred
	,citecolor=cblue
}

\usepackage[
    backend=biber,
    style=alphabetic,
    labelnumber,
    defernumbers=true,
    giveninits=true,
    maxbibnames=299,
    natbib=true,
    doi=false,
    isbn=false,
    url=false,
    sorting=nyt,
    date=year,
]
{biblatex}
\bibliography{higher_dim}

\renewcommand{\d}{\mathrm{d}}
\newcommand{\D}{\mathcal{D}}
\renewcommand{\H}{\mathcal{H}}
\newcommand{\R}{\mathbb{R}}

\DeclareMathOperator{\dom}{dom}

\newcommand{\norm}[1]{\left\Vert #1 \right\Vert}
\newcommand{\argdot}{{\hspace{0.18em}\cdot\hspace{0.18em}}}
\newcommand{\scalar}[2]{\langle #1, #2\rangle}
\renewcommand{\Re}{\mathrm{Re}\,}
\renewcommand{\Im}{\mathrm{Im}\,}

\newtheorem{theorem}{Theorem}[section]
\newtheorem{corollary}[theorem]{Corollary}
\newtheorem{lemma}[theorem]{Lemma}
\newtheorem{proposition}[theorem]{Proposition}
\theoremstyle{definition}
\newtheorem{definition}[theorem]{Definition}

\theoremstyle{remark}
\newtheorem{remark}[theorem]{Remark}

\numberwithin{equation}{section}
\newtheorem{example}{Example}
\newcommand{\triang}{\hfill$\triangle$}

\title[Quantum controllability on graph-like manifolds]{}

\subjclass[2010]{81Q93, 81Q35, 35Q41, 35J10}

\begin{document}
	
		\begin{center}
		\LARGE
    Quantum controllability on graph-like manifolds \\
    through magnetic potentials and boundary conditions
		\par \bigskip
		
		\normalsize
		\authorfootnotes
		Aitor Balmaseda\footnote{abalmase@math.uc3m.es}\textsuperscript{,1,2},
		Davide Lonigro\footnote{davide.lonigro@ba.infn.it}\textsuperscript{,3,4,5}, and
		Juan Manuel Pérez-Pardo\footnote{jmppardo@math.uc3m.es, corresponding author}\textsuperscript{,1,6} \par \bigskip
		
		\textsuperscript{1}\footnotesize Departamento de Matemáticas, Universidad Carlos III de Madrid, Avda.\ de la Universidad 30, 28911 Madrid, Spain \par
		\textsuperscript{2}\footnotesize Departamento de Análisis Matemático y Matemática Aplicada, Facultad de Ciencias Matemáticas, Universidad Complutense de Madrid, Madrid 28040. Spain\par
		\textsuperscript{3}\footnotesize Dipartimento di Matematica, Università di Bari, I-70125 Bari, Italy \par
		\textsuperscript{4}\footnotesize Dipartimento di Fisica, Università di Bari, I-70126 Bari, Italy \par
		\textsuperscript{5}\footnotesize INFN, Sezione di Bari, I-70126 Bari, Italy \par
		\textsuperscript{6}\footnotesize Instituto de Ciencias Matemáticas (CSIC - UAM - UC3M - UCM) ICMAT, C/ Nicolás Cabrera 13--15, 28049 Madrid, Spain
		\par \bigskip
		
		\today
	\end{center}

  \begin{abstract}
We investigate the controllability of an infinite-dimensional quantum system: a quantum particle confined on a Thick Quantum Graph, a generalisation of Quantum Graphs whose edges are allowed to be manifolds of arbitrary dimension with quasi-$\delta$ boundary conditions.
This is a particular class of self-adjoint boundary conditions compatible with the graph structure.
We prove that global approximate controllability can be achieved using two physically distinct protocols: either using the boundary conditions as controls, or using time-dependent magnetic fields.
Both cases have time-dependent domains for the Hamiltonians.
    \end{abstract}

\maketitle
\vspace{-1cm}

%%%%%%%%%%%%%%%%%%%%%%%%%%%%%%%%%%%%%%%%%%%%%%%%%%%%%%%%%%%%%%%%%%%%%%%%%%%%%%%%%%%%%%%%%%%%%%%%%%%%%%%%%%%%%%%%%%%%%%%%%%%%%%%%%%%%%%%%%%%%%%%%%%%%%%%%%%%%%%%%%%%%%%%%%%%%%%%%%%%%%%%%%%%%%%%%%%%%%%%%%%%%%%%%%%%%%%%%%%%%%%%%%%%%%%%%%%%%%%%%%%%

\section{Introduction}

The theory of Quantum Control has had a great success when applied to finite-dimensional quantum systems and finite-dimensional approximations of infinite-dimensional systems.
This success is, partly, due to the development of the geometric theory of control \cite{AgrachevSachkov2004,Jurdjevic1996}{, which applies to finite-dimensional quantum systems.
In \cite{DAlessandro2007} there is a good account on quantum control techniques, including optimal control theory, applied to the finite-dimensional case.
}
The transition from finite-dimensional to infinite-dimensional control {(that is, from control of dynamical systems with a finite-dimensional space of states to infinite-dimensional ones)} is not an easy task.
In fact, many of the results that apply for finite-dimensional quantum systems do not carry over to infinite-dimensional ones.
A remarkable example of this is the notion of exact controllability, which is not suitable for infinite-dimensional systems \cite{BallMarsdenSlemrod1982,Turinici2000,IbortPerezPardo2009}.
Most of the quantum systems that appear in the applications to develop quantum computers, like ion traps or superconducting circuits, are infinite-dimensional.
However, the models used to describe them are finite-dimensional approximations \cite{Wendin2017,LawEberly1996}.
This introduces errors in their description, as an interaction of the system with itself must be neglected and forces to work in the lowest energy levels of the system to avoid such unwanted interactions.
It is therefore natural to look for better mathematical models for them.

{
The theory of control of finite-dimensional quantum systems is well developed and can be considered mature, including the development of optimal control techniques in the finite-dimensional case, see for instance the recent reviews \cite{KochBoscainCalarcoEtAl2022,GlaserBoscainCalarcoEtAl2015}.
As accounted there, the control of infinite-dimensional quantum systems is at an earlier stage and there are not yet robust general characterisations of controllability.
However, this research field is attracting more and more attention.
}

{
During the last decade, relevant results have been obtained regarding the controllability of infinite-dimensional quantum systems.
For instance, results for adiabatic control \cite{AugierBoscainSigalotti2022,DucaJolyTuraev2023,IbortMarmoPerezPardo2014a,RobinAugierBoscainEtAl2022} or for bilinear quantum control systems, e.g.\ \cite{Beauchard2005a, beauchard2010local, beauchard2010controllability, mirrahimi2009lyapunov, nersesyan2010global}.
Most of these results are obtained on a case-by-case basis.
}
Of particular interest are so-called Lie-Galerkin techniques that have been obtained and successfully applied, for instance, to the control of molecules \cite{ChambrionMasonSigalottiEtAl2009, BoscainCaponigroChambrionEtAl2012, BoussaidCaponigroChambrion2019, BoscainPozzoliSigalotti2021}.
The latter results apply under fairly general conditions and the family of controls considered are piecewise constant functions.
This is a simplifying hypothesis that allows one to bypass the problem of existence of solutions of the time-dependent Schrödinger equation.

{
However, piecewise constant controls are not suitable to address certain situations, for instance \cite{ErvedozaPuel2009} or when the domains of the Hamiltonian operators are not constant.
}
A remarkable particular case is when time-dependent magnetic fields are considered, as the boundary conditions determining {self-adjoint Hamiltonians are typically time-dependent in that case.}
One of the purposes of this article is to take a step forward in addressing this problem.
By using weak solutions of the Schrödinger equation we are able to deal with Hamiltonians with time-dependent domains.

Most results in infinite-dimensional control theory rely on external fields, but an alternative and physically appealing option is to seek controllability by manipulating the \textit{boundary} of the system, for example by modifying the boundary conditions.
Similarly, an option for controlling a system with boundary is to actively deform and/or translate the region in which the system is confined.
The simplest example of such a framework is a quantum particle confined in a one-dimensional box with moving walls (see, e.g., \cite{DiMartinoAnzaFacchiEtAl2013}).
The possibility of controlling such a system by manipulating the motion of the walls was investigated in \cite{BalmasedaLonigroPerezPardo2022,BeauchardCoron2006a,Rouchon2003}, where controllability between eigenstates of a quantum particle via a rigidly moving potential well is proven, and in \cite{CarrascoRoganValdivia2017,DuffinDijkstra2019}, where the case of an oscillating well is numerically investigated, with the main focus again being controllability between eigenstates.
The results presented in this article set up a framework to deal with this kind of systems and to study their controllability.

In this work we present a scheme for controlling the state of infinite-dimensional quantum systems whose dynamics are governed by Schrödinger equations in which the self-adjoint domains of the Hamiltonian operators are possibly time-dependent.
This allows us to prove approximate controllability results using generic magnetic fields as controls and also to introduce a new scheme of quantum control by using time-dependent boundary conditions as controls.
This control scheme does not rely on the action of external fields to drive the state of the system.
The aim is to minimise the interaction with the system hoping that this allows for longer coherent times in the manipulation of quantum systems.
This is a complete new type of control of quantum systems that is possible only in the context of infinite-dimensional quantum systems.
It is worth to stress that the type of control by changing the boundary conditions that we present here is different from other schemes of boundary control studied in the literature during the last 50 years, cf.\ \cite{Fattorini1968, Lions1971}, where the value of the solution at the boundary is prescribed at given times.

{
  As a generic set of models to present this theory, we introduce a class of quantum systems that we call \emph{Thick Quantum Graphs}.
  }
  Thick Quantum Graphs are a generalisation of Quantum Graphs, i.e., metric graphs equipped with a differential operator defined on its edges and appropriate boundary conditions to determine a self-adjoint extension, c.f.\ \cite{KostrykinSchrader2003,Kuchment2004}; where the edges can be manifolds of any dimension.
  Quantum Graphs are thus a particular case of Thick Quantum Graphs and, therefore, the controllability results contained in this article also apply to them.
  {It is worth to stress that} the controllability results proven in this work are also novel from the point of view of Quantum Graphs.

  {
  The purpose of our notion of Thick Quantum Graphs is to describe a quantum particle moving on a circuit-like space made out a number of wires connected according to a network structure.
  Let us remark now that our construction is different from other generalisations of Quantum Graphs with similar names such as ``fat graphs'', ``thick graphs'', ``graph-like spaces'', ``graph neighbourhoods'', ``tubular branched manifolds'', etc.
  Those notions are well suited when considering approximations of thin branched manifolds, while Thick Quantum Graphs are meant to describe a network of wires with an arbitrary shape (see Section~\ref{sec:circuits} for further discussion).
  }%
  For a comprehensive and detailed explanation of these ideas, we refer to \cite{Post2012} and references therein.
  Another related concept is that of Quantum Waveguides, which correspond to Fat Quantum Graphs with Dirichlet boundary conditions \cite{ExnerPost2005,ExnerPost2007,KrejcirikRaymond2014,KrejcirikRaymondTusek2015,MolchanovVainberg2006}.
  {These constructions, being focused in studying how thin network-shaped systems converge to Quantum Graphs, do not provide an easy approach to consider general vertex conditions, i.e., general boundary conditions, which motivates the construction of Thick Quantum Graphs.
  }

  {
  Quantum Graphs and their generalisations provide a way to model different physical systems potentially relevant for the development of quantum technologies.
  }%
  For instance, any vertex coupling on a Quantum Graph can be approximated by Schrödinger operators on thin manifolds \cite{ExnerPost13}, a Quantum Graph model of the anomalous Hall effect has been proposed recently \cite{StredaKucera2015, StredaVyborny2022}, {and there are applications to }thin superconducting networks \cite{Alexander1985, rubinstein1998multiply, rubinstein2006}, quantum wires circuits \cite{exner1989electrons} or photonic crystals \cite{kuchment2002differential}.
  More directly connected with the research presented in this article are the possible applications to quantum computation based on them \cite{cheon2004quantum, tanaka2007quasienergy, tanaka2010adiabatic}.
  There are simply too many examples in the literature to give a comprehensive account.
  Instead, we refer to the thorough overview by P.~Exner in Appendix~K of \cite{AlbeverioGesztesyHoeghKrohnEtAl2005}.
  Besides their relevance, time-dependent problems and, in particular, control problems on these models have been scarcely addressed \cite{BalmasedaPerezPardo2019, DellAntonioFigariTeta2000}.
  This justifies a thorough analysis of the dynamics and controllability on Thick Quantum Graphs.
{
  Some explicit examples of their construction can be found in Section~\ref{sec:preliminaries} (see Examples~\ref{ex:basicTQGexample}, \ref{ex:1d-example} and \ref{ex:2d-example}) and at the end of Section~\ref{sec:TQG} (cf.\ Figure~\ref{fig:examples}).
}

  Motivated by these problems{,} we focus on two selected families of control problems on Thick Quantum Graphs that involve magnetic Laplacians and related boundary conditions.
  {First, }we use time-varying magnetic fields to control the state of the system while we take into account the induced electric fields.
  {Secondly,} we consider the situation where the state of the system is driven by changing the boundary conditions.
  We prove, under certain assumptions, that these systems are approximately controllable.
  In order to address the problem of controllability we need, in particular, to address the problem of the existence of solutions of the Schrödinger equation with time-dependent boundary conditions, a problem that has its own significance.
  Situations similar to the ones presented here have been considered before in some particular quantum systems \cite{DellAntonioFigariTeta2000,DiMartinoAnzaFacchiEtAl2013,PerezPardoBarberoLinanIbort2015}.

  {
  The Hamiltonians considered on this work contain either Laplacians or magnetic Laplacians defined on manifolds with boundary.
  There are different approaches to characterise the self-adjoint extensions of such Laplacians: quasi-boundary triples \cite{BehrndtLanger2012}, the formalism of boundary pairs for positive sesquilinear forms \cite{Post2016} or the use of pseudo-differential operators that allow one construct a boundary triple \cite{Grubb1968}.
  Instead, we will follow the approach in \cite{IbortLledoPerezPardo2015}, which we extended from the case of the Laplace-Beltrami operator to the case of magnetic Laplacians (see Section~\ref{subsec:sa_extensions_quad_forms} for further details).
  }%
  {In this approach, self-adjoint extensions are parametrised} in terms of unitary operators acting on the space of boundary data.
  These unitaries are suitable to carry representations of the symmetry groups of the underlying manifold \cite{IbortLledoPerezPardo2015a}.
  This fact was used in \cite{BalmasedaDiCosmoPerezPardo2019} to characterise the boundary conditions compatible with the graph structure that we consider here.
  Finally, let us remark that the characterisation of self-adjoint extensions that we use in this work can be used for other differential operators like Dirac operators \cite{IbortPerezPardo2015,PerezPardo2017,IbortLlavonaLledoEtAl2021}.
  Therefore, it is likely that similar results can be proven also in these other situations.
  Moreover, this characterisation is also suitable for numerics \cite{IbortPerezPardo2013,LopezYelaPerezPardo2017}, so that further developments of the control problem with applications to optimal control may also be enabled by the results presented here.

{
  The main contributions of this work are a characterisation of controllability of a quantum system by modifying the boundary conditions, a characterisation of the controllability of charged particles by magnetic fields that takes into account the induced electric field and a new set of models, Thick Quantum Graphs, that can accommodate the above description in dimension higher than one and include Quantum Graphs as a particular case.
  The results about controllability are also new in this latter case.
  In order to prove those results, we have generalised the characterisation of self-adjoint extensions in \cite{IbortLledoPerezPardo2015} to the case of magnetic Laplacians.
  }

This article is organised as follows.
In Section~\ref{sec:preliminaries} some {illustrative examples} are {introduced} and the main results of this work are presented.
The self-adjoint extensions of the magnetic Laplacians on generic manifolds with boundary are {characterised} in Section~\ref{sec:magneticLaplacian} and {our} construction of Thick Quantum Graphs is given in Section~\ref{sec:TQG}.
In Section~\ref{sec:circuits} we prove controllability results on Thick Quantum Graphs, cf.\ Theorem~\ref{thm:controllability_piecewise_smooth} and Theorem~\ref{thm:controllability_boundary}.
Finally, in Section~\ref{sec:examples} we show how results of the previous sections can be used to extend controllability results to other meaningful situations.

%%%%%%%%%%%%%%%%%%%%%%%%%%%%%%%%%%%%%%%%%%%%%%%%%%%%%%%%%%%%%%%%%%%%%%%%%%%%%%%%%%%%%%%%%%%%%%%%%%%%%%%%%%%%%%%%%%%%%%%%%%%%%%%%%%%%%%%%%%%%%%%%%%%%%%%%%%%%%%%%%%%%%%%%%%%%%%%%%%%%%%%%%%%%%%%%%%%%%%%%%%%%%%%%%%%%%%%%%%%%%%%%%%%%%%%%%%%%%%%%%%%

\section{Mathematical Preliminaries and Main Results}\label{sec:preliminaries}

  In this section we will present the main results of this work, Theorem~\ref{thm:controllability_piecewise_smooth} and Theorem~\ref{thm:controllability_boundary}, together with some meaningful examples to which they can be applied to.
  Before doing so we need to introduce several useful notions.
  As stated in the introduction, we will study the feasibility of controlling certain quantum systems using their boundary conditions as controls.
  First, one needs to address the problem of the existence of solutions for the time-dependent and non-autonomous Schrödinger equation on these systems, the latter being a linear evolution equation on a complex separable Hilbert space $\mathcal{H}$ defined by a family of self-adjoint operators $\{H(t)\}_{t \in\R}$, densely defined on subsets of $\mathcal{H}$, called the time-dependent Hamiltonian.
  To simplify the notation we will denote this family as $H(t)$, $t\in I$.

  In the most general case, the time-dependent Hamiltonian is a family of unbounded self-adjoint operators whose domains depend on $t$.
  These facts make the problem of the existence of solutions for the Schrödinger equation highly non-trivial.
  There are general sufficient conditions for the existence of solutions of these equations, see for instance \cite{Kato1973,Kisynski1964,Simon1971}.
  An important notion to address this problem is the notion of sesquilinear form associated with an operator.
  For details about them and the associated representation theorems we refer to \cite[Sec.~VI.2]{Kato1995}.

  The following families of time-dependent Hamiltonians define the evolution problems that we are interested in.
  \begin{definition} \label{def:hamiltonian_const_form}
    Let $I \subset \mathbb{R}$ be a compact interval, let $\mathcal{H}^+$ be a dense subspace of $\mathcal{H}$ and $H(t)$, ${t \in I}$, a family of self-adjoint operators on $\mathcal{H}$, where for any $t\in I$ the operator $H(t)$ is densely defined on $\D(t)$.
    We say that $H(t)$, ${t \in I}$, is a \emph{time-dependent Hamiltonian with constant form domain} $\mathcal{H}^+$ if:
    \begin{enumerate}[label=\textit{(\roman*)},nosep,leftmargin=*]
      \item There is $m >0$ such that, for any $t\in I$, $\langle \Phi, H(t)\Phi \rangle \geq -m \|\Phi\|^2$ for all $\Phi\in\D(t)$.
      \item For any $t \in I$, the domain of the Hermitian sesquilinear form $h_t$ associated with $H(t)$ by the representation theorem, i.e.\ the form domain, is $\mathcal{H}^+$.
    \end{enumerate}
  \end{definition} 
    \begin{definition}\label{def:unitarypropagator} A \emph{unitary propagator} is a two-parameter family of unitary operators $U(t,s)$, $s,t\in\R$ that satisfies:
  \begin{enumerate}[label=\textit{(\roman*)},nosep,leftmargin=*]
     \item $U(t,s)U(s,r)= U( t , r)$.
     \item $U(t,t)= \mathbb{I}$.
     \item $U(t,s)$ is jointly strongly continuous in $t$ and $s$.
  \end{enumerate}
  \end{definition}
  Next we introduce several notions of solutions of the Schrödinger equation.

  \begin{definition}\label{def:strong_Schrodinger}
  Let $I\subset\R$ be a compact interval and $H(t)$, $t\in I$, be a time-dependent Hamiltonian.
  Consider the equation
  \begin{equation}\label{eq:schro0}
    i\frac{\mathrm{d}}{\mathrm{d}t}\Psi(t)=H(t)\Psi(t).
  \end{equation}
  We say that a unitary propagator $U(t,s)$, $t,s\in I$, is a \textit{(strong) solution of the Schrödinger equation} if, for all $\Psi_0\in\mathcal{D}(H(s))$, the function $t\in I \mapsto \Psi(t):=U(t,s)\Psi_0$ solves Eq.~\eqref{eq:schro0} with initial condition $\Psi(s)=\Psi_0$.
  \end{definition}

  \begin{definition}\label{def:weak_Schrodinger}
  Let $I\subset\R$ be a compact interval and $H(t)$, $t\in I$, be a time-dependent Hamiltonian with constant form domain $\H^+$,
  and let $\Phi\in\H^+$.
    Consider the equation	
  \begin{equation}\label{eq:schro_weak}
    i\frac{\mathrm{d}}{\mathrm{d}t}\scalar{\Phi}{\Psi(t)}=h_t\left(\Phi,\Psi(t)\right),
  \end{equation}
  with $h_t(\cdot,\cdot)$ being the sesquilinear form uniquely associated with $H(t)$.
    We say that a unitary propagator $U(t,s)$, $t,s\in I$, is a \textit{weak solution of the Schrödinger equation} if, for $\Psi_0\in\H^+$, the function $t\in I\mapsto \Psi(t):=U(t,s)\Psi_0$ solves Eq.~\eqref{eq:schro_weak} with initial condition $\Psi(0)=\Psi_0$ for all $\Phi\in\H^+$.
  \end{definition}

  For quantum control purposes, we will also need an even weaker notion of solution that considers propagators solving Eq.~\eqref{eq:schro_weak} for all but finitely many values of $t$, that is, admitting finitely many time singularities, see e.g.~\cite{BoscainCaponigroChambrionEtAl2012,boussaid2013weakly,ChambrionMasonSigalottiEtAl2009}.
  This leads us to the following definition:
  \begin{definition}\label{def:piecewise-sol}
  Let $I\subset\R$ be a compact interval and $H(t)$, $t\in I$, be a time-dependent Hamiltonian with constant form domain $\H^+$.
  We say that a unitary propagator $U(t,s)$, $t,s\in I$, is a \textit{{piecewise weak} solution of the Schrödinger equation} if, for all $\Psi_0\in\H^+$, there exist $t_0<t_1<\dots<t_d\in I$ and a family of weak solutions of the Schrödinger equation $\{U_i(t,s) \mid t,s \in (t_{i-1},t_i)\}_{i=1,\dots,d}$ such that for $t\in(t_{i-1},t_i)$, $s\in(t_{j-1},t_j)$, $1\leq j<i \leq d$ the unitary propagator $U(t,s)$ can be expressed as
  \begin{equation}
    U(t,s)=U_i(t,t_{i-1})U_{i-1}(t_{i-1},t_{i-2})\cdots U_j(t_{j},s).
  \end{equation}
  \end{definition}

  It is straightforward to prove that strong solutions of the Schrödinger equation are also weak solutions and that the latter are also {piecewise weak} solutions.
  Several results of this article are obtained for families of controls that are piecewise differentiable.
  The precise definition of these sets of functions is given next.

  \begin{definition} \label{def:piecewise_differentiable}
    We will say that a function $f$ is $n$-times piecewise differentiable on $I \subset \mathbb{R}$, denoted $f \in {C_\mathrm{{pw}}^n(I)}$, if there exists a finite collection of open subintervals of $I$, $\{I_j\}_{j=1}^\nu$, such that they are pairwise disjoint with $I = \bigcup_{j} \overline{I_j}$, and there exists a collection of $n$-times differentiable functions $\{g_j\}_{j=1}^\nu \subset C^n(I)$ satisfying $f|_{I_j} = g_j|_{I_j}$ for $j = 1, \dots, \nu$.
  \end{definition}

  The rest of this section is devoted to present the particular families of quantum systems that we consider in this work, and their controllability properties.
  We consider two families of control problems defined on Thick Quantum Graphs, cf.\ Definition~\ref{def:thick-quantum-graph}.

{
  Our notion of Thick Quantum Graph aims to describe the situation of a quantum particle confined to move on a circuit-like space.
  For that purpose, the main ingredients of our description must be the wires of the circuits and its connections.
  The structure of the circuit can be depicted by a graph whose edges represent the wires and whose vertices represent the connections between wires.
  The shape of each wire is given by a Riemannian manifold with boundary, which is associated with each of the edges.
}

  {Therefore,} Thick Quantum Graphs are graph-like manifolds of any dimension and are denoted by a triple $(G,\Omega, \Gamma)$, where $G$ is a graph, $\Omega$ is a non-connected differentiable manifold and $\Gamma$ is a partition of the boundary of the manifold $\Omega$.
  We will denote by $V$ the vertex set of the graph and by $E$ the edges set.
  For $v\in V$ the set $E_v$ will denote the set of edges that share the vertex $v$.

  We consider different self-adjoint extension of the Laplace-Beltrami operator and of magnetic Laplacians defined on Thick Quantum Graphs.
  These extensions are going to be determined by a family of boundary conditions that we call quasi-$\delta$ type boundary conditions.
  These are a natural generalisation of self-adjoint boundary conditions on Quantum Graphs.
  We leave the precise details of the definition of Thick Quantum Graphs and of quasi-$\delta$ type boundary conditions to Section~\ref{sec:TQG}.
  Magnetic Laplacians, cf.\ Definition~\ref{def:magnetic-laplacian}, are defined in Section~\ref{sec:magneticLaplacian}.

  Thick Quantum Graphs contain Quantum Graphs as a particular case, cf.\ \cite{BerkolaikoKuchment2012}; therefore, the results and theorems presented below and in Section~\ref{sec:circuits}, that apply to the general context of Thick Quantum Graphs, also hold in the one-dimensional situation represented by Quantum Graphs.
  It is worth to stress that the results that we present are also novel in the context of Quantum Graphs.

  {
  Instead of giving at this point the precise definition of Thick Quantum Graphs, which is given in Definition~\ref{def:thick-quantum-graph}, we present some meaningful illustrative examples.
}
  
{
  \begin{example}\label{ex:basicTQGexample}
    For a concrete example consider the one shown on Figure~\ref{fig:TQGexample}, where both a one-dimensional Thick Quantum Graph (that is, a Quantum Graph), and a two-dimensional Thick Quantum Graph are depicted (Subfigures~\ref{subfig:TQG-graph} and \ref{subfig:TQG-wire} respectively).
    Both examples share the same structure: a single wire connected to itself.
    Therefore, both Thick Quantum Graphs share the graph describing their structure (which coincides with the one depicted on Figure~\ref{subfig:TQG-graph}), while the shape of the wires itself is what differentiates both cases.
    For the one-dimensional case, Fig.~\ref{subfig:TQG-graph}, the manifold describing the wire is just one interval.
    On the other hand, for the case depicted on Fig.~\ref{subfig:TQG-wire}, the manifold describing the wire is a bended two-dimensional strip.\triang
  \end{example}
  \begin{remark}
    It is worth stressing the fact that Quantum Graphs are included inside our notion of Thick Quantum Graphs with one-dimensional manifolds as wires.
    Also the standard notion of Fattened Graphs is contained in our definition, since it corresponds with wire manifolds consisting on (usually tubular) neighbourhoods of the associated graph.
    The notion presented here is, however, more general since it allows arbitrary shapes for the related manifolds (see Section~\ref{sec:TQG} for further discussion).
  \end{remark}
}

\renewcommand{\thesubfigure}{\Alph{subfigure}}
  \begin{figure}[b]
    \centering
    \begin{subfigure}[b]{0.45\textwidth}
      \centering
      \includegraphics[scale=1]{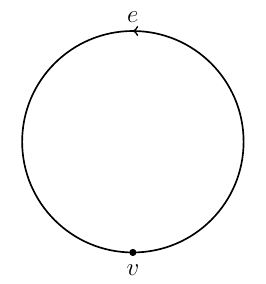}
      \subcaption{{One-dimensional Thick Quantum Graph.}}
      \label{subfig:TQG-graph}
    \end{subfigure}
    \begin{subfigure}[b]{0.45\textwidth}
      \centering
      \includegraphics[scale=1]{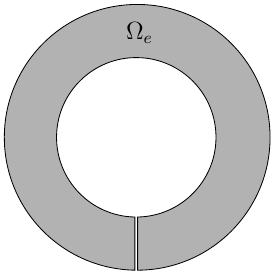}
      \subcaption{{Two-dimensional Thick Quantum Graph.}}
      \label{subfig:TQG-wire}
    \end{subfigure}
    \caption{
      {
        Two examples of Thick Quantum Graph whose wires have different dimensions but both having the same network structure: a wire connected to itself.
      }
    }
    \label{fig:TQGexample}
  \end{figure}

{
  The main results on this work apply for two families of control systems defined on Thick Quantum Graphs: \emph{quantum induction control systems} (Definition~\ref{def:induction_control_system}) and \emph{quasi-$\delta$ boundary control systems} (Definition~\ref{def:boundary_control_system}).
  Let us now introduce the relevant notation and ideas through two examples using different Thick Quantum Graphs, while keeping the precise definitions for latter sections.

\begin{example}\label{ex:1d-example}
  Consider the Thick Quantum Graph $\Omega$ depicted on Figure~\ref{subfig:TQG-graph}, consisting of a quantum graph with a single edge of length $2\pi$ and a single vertex (the identification of $\{0\}\sim\{2\pi\}$).
  We can define the following control systems with time-dependent self-adjoint boundary conditions, i.e., time-dependent self-adjoint domain, for the Laplace operator and the magnetic Laplacian:
  \begin{enumerate}[label=\textit{\alph*}),nosep,leftmargin=*]
    \item Consider the quantum dynamical system
      \begin{equation*}
        i \frac{\d}{\d t} \Psi(t) = - \left(\frac{\d}{\d x} + i a(t) A_0\right)^2 \Psi(t) + a'(t) A_0 x \Psi(t),
      \end{equation*}
      with $\Psi(t) \in \bigl\{\Phi \in H^2(\Omega) \mid \Phi(0) = \Phi(2\pi), \Phi'(0) = \Phi'(2\pi)\bigr\}$, $H^2(\Omega)$ being the Sobolev space of order 2 on the Thick Quantum Graph, and $A_0 \in \mathbb{R}$.
      This system represents, assuming natural units, a quantum particle on a spire of radius $1$, that encircles a solenoid that generates a magnetic flux of magnitud $2\pi a(t)A_0$ through the spire \cite{PerezPardoBarberoLinanIbort2015}.
      The term $a'(t)xA_0$ is the electric potential associated with Faraday's induction law corresponding to the time-varying magnetic flux determined by the magnetic potential $a(t)A_0$.
      This kind of control problem, using the function $a(t)$ as control, is an instance of what we call quantum induction control system in this work.
    \item Consider the quantum dynamical system
      \begin{equation*}
        i\frac{\d}{\d t}\Psi(t) = -\frac{\d^2}{\d x^2}\Psi(t),
      \end{equation*}
      with $\Psi(t)\in\bigl\{ \Psi\in\H^2([0,2\pi]) \mid \Psi(0) = e^{i\chi(t)}\Psi(2\pi), \Psi'(0) = e^{i\chi(t)}\Psi'(2\pi) \bigr\}.$
      This kind of control problem, using the function $\chi(t)$ as control, is an instance of what we call quasi-$\delta$ boundary control problem.
    \triang
  \end{enumerate}
\end{example}
}

{
\begin{example}\label{ex:2d-example}
  We consider now the two-dimensional Thick Quantum Graph on Fig.~\ref{subfig:TQG-wire}.
  Once again we have a single wire connected to itself.
  However, in this case the wire is a bended two-dimensional strip.
  We use polar coordinates $\theta \in [0, 2\pi]$ and $r \in [a, b]$, and define the following systems:
  \begin{enumerate}[label=\textit{\alph*}),nosep,leftmargin=*]
    \item Consider the quantum dynamical system
      \begin{equation*}
        i \frac{\d}{\d t} \Psi(t) = \left[- \frac{\partial^2}{\partial r^2} - \frac{1}{r} \frac{\partial}{\partial r} - \frac{1}{r^2} \left( \frac{\partial}{\partial \theta} + ia(t)A_0\theta \right)^2 + a'(t) A_0 \theta\right] \Psi(t),
      \end{equation*}
      with $\Psi(t) \in \bigl\{ \Phi \in H^2(\Omega) \mid \Phi(a,\theta) = \Phi(b,\theta) = 0, \Phi(r,0) = \Phi(r, 2\pi), \frac{1}{r} \frac{\partial \Phi}{\partial \theta}(r, 0) = \frac{1}{r} \frac{\partial\Phi}{\partial \theta}(r, 2\pi) + \delta \Phi(r, 0) \bigr\}.$
      As before, the term $a'(t)A_0\theta$ is the electric potential associated with Faraday's induction law.
      This system, with the function $a(t)$ as control, is another instance of a quantum induction control system.
    \item Consider the dynamical system
      \begin{equation*}
        i \frac{\d}{\d t} \Psi(t) = - \left( \frac{\partial^2}{\partial r^2} + \frac{1}{r} \frac{\partial}{\partial r} + \frac{1}{r^2} \frac{\partial^2}{\partial \theta^2} \right) \Psi(t),
      \end{equation*}
      with $\Psi(t) \in \bigl\{\Phi \in H^2(\Omega) \mid \Phi(a,\theta) = \Phi(b,\theta) = 0, e^{i\chi(t)}\Phi(r, 2\pi) = \Phi(r,0), \frac{e^{i\chi(t)}}{r} \frac{\partial\Phi}{\partial \theta}(r, 2\pi) = \frac{1}{r} \frac{\partial\Phi}{\partial\theta}(r, 0) + \delta \Phi(r,0)\bigr\}$.
      This kind of control system, using the function $\chi(t): I \to \mathbb{R}$ as control and $\delta \in \mathbb{R}$ a fixed parameter, is another instance of a quasi-$\delta$ boundary control system.
    \triang
  \end{enumerate}
\end{example}

}

The following controllability theorems on systems with time-dependent domains constitute the main contributions of this work.
In particular (Theorem~\ref{thm:controllability_boundary}), we are able to prove controllability by using the space of self-adjoint extensions, i.e., the boundary conditions, as controls.

{
\begin{theorem}\label{thm:controllability_piecewise_smooth}
  Let $H(t)$ be the Hamiltonian of a quantum induction control system (cf.\ Def.~\ref{def:induction_control_system}).
  Let $\Psi_0, \Psi_T\in\H$ with $\norm{\Psi_0} = \norm{\Psi_T}$ and $r\in\R$ a positive number.
  Then, for all $\varepsilon>0$ there exists $T>0$ and a control $u\in C_\mathrm{{pw}}^1([0,T])$ with $|u'|<r$ almost everywhere such that the controlled Schrödinger equation with Hamiltonian $H(t)$ has a {piecewise weak solution} $\Psi(t)$ satisfying
    
  $$\Psi(0)=\Psi_0 \text{ and }\norm{\Psi_T-\Psi(T)} < \varepsilon.$$
\end{theorem}

\begin{theorem}\label{thm:controllability_boundary}
  Let $H(t)$ be the Hamiltonian of a quasi-$\delta$ boundary control system (cf.\ Def.~\ref{def:boundary_control_system}).
  Let $\Psi_0, \Psi_T\in\H$, with $\norm{\Psi_0} = \norm{\Psi_T}$, $r\in\R$ a positive number and $u_0, u_1\in\R$.
  Then, for all $\varepsilon>0$ there exists $T>0$ and a control $u\in C_\mathrm{{pw}}^1([0,T])$ with $|u'|<r$ almost everywhere, $u(0)=u_0$ and $u(T)=u_1$ such that the controlled Schrödinger equation with Hamiltonian $H(t)$ has a {piecewise weak solution} $\Psi(t)$ satisfying

  $$\Psi(0)=\Psi_0 \text{ and }\norm{\Psi_T-\Psi(T)} < \varepsilon.$$
\end{theorem}
}

Theorem~\ref{thm:controllability_piecewise_smooth} and Theorem~\ref{thm:controllability_boundary} establish that the dynamical systems of Example~\ref{ex:1d-example} and Example~\ref{ex:2d-example} are approximately controllable, see Definition~\ref{strong_controllability}, with controls given in the respective theorems and whose solutions of the Schrödinger equation are {piecewise weak solutions}.
In Section~\ref{sec:examples} we present some further examples, being either direct applications or straightforward modifications of these theorems.

%%%%%%%%%%%%%%%%%%%%%%%%%%%%%%%%%%%%%%%%%%%%%%%%%%%%%%%%%%%%%%%%%%%%%%%%%%%%%%%%%%%%%%%%%%%%%%%%%%%%%%%%%%%%%%%%%%%%%%%%%%%%%%%%%%%%%%%%%%%%%%%%%%%%%%%%%%%%%%%%%%%%%%%%%%%%%%%%%%%%%%%%%%%%%%%%%%%%%%%%%%%%%%%%%%%%%%%%%%%%%%%%%%%%%%%%%%%%%%%%%%%%%%%%%%%%%%%%%%%%%%%%

\section{Self-adjoint extensions of the Laplacian and the Magnetic Laplacian}\label{sec:magneticLaplacian}

 In this work we consider the particular case of time-dependent Hamiltonians defined by self-adjoint extensions of the Laplace-Beltrami operator on Riemannian manifolds and related differential operators like magnetic Laplacians.
 We will now present some important results concerning this class of examples.

\subsection{Sobolev spaces and the Laplace-Beltrami operator}
 
  Let $\Omega$ be a Riemannian manifold with boundary $\partial \Omega$ and let $\eta$ denote the Riemannian metric.
  We will consider only orientable, compact and smooth manifolds with piecewise smooth boundary{, that is $\partial\Omega = \cup_{i=1}^N\overline{\partial\Omega_i}$ where each $\partial\Omega_i$ is an open submanifold of $\partial\Omega$ and is itself a smooth manifold with smooth boundary.}
  We will assume that the boundary is Lipschitz, and will denote by $\Lambda(\Omega)$ the space of smooth, complex-valued 1-forms on $\Omega$.
 
  For $\alpha, \beta\in \Lambda(\Omega)$ smooth, complex-valued 1-forms in $\Omega$, we define the inner product
  \begin{equation*}
    \langle \alpha, \beta \rangle_{\Lambda} = \int_\Omega \eta^{-1}(\overline{\alpha}, \beta) \, \d\mu_\eta,
  \end{equation*}
  and $\|\cdot\|_{\Lambda}$ the associated norm, which allows for defining the space of differential 1-forms with square-integrable coefficients.
  For the rest of the work we will drop the subindex in the norm and inner product defined above whenever it leads to no confusion.

  We denote by $H^k(\Omega)$ the Sobolev space of class $k$ on $\Omega$ (see \cite{Taylor2011,LionsMagenes1972}), and its norm $\|\Phi\|^2_{H^k(\Omega)}$.
  When there is no risk of confusion about what the manifold $\Omega$ we are working with is, we will simplify the notation and write $\|\cdot\|_k \coloneqq \|\cdot\|_{H^k(\Omega)}$.

  The Laplace-Beltrami operator can be defined on the space of smooth functions $C^{\infty}(\Omega)$ in terms of the exterior differential, $\d$, and the codifferential, $\delta$, cf.\ \cite{MarsdenAbrahamRatiu1988},
  \begin{equation*}
    \Delta = \delta\d.
  \end{equation*}
  For the rest of this work all the derivatives acting on the manifold $\Omega$ are going to be considered in the weak sense.
The boundary $\partial\Omega$ of the Riemannian manifold inherits the structure of a Riemannian manifold with the metric given by the pull-back of the Riemannian metric on $\Omega$, which we will denote by $\partial\mu$.

In what follows we will denote by $\gamma: H^k(\Omega) \to H^{k-1/2}(\partial\Omega)$, $k > 1/2$, the continuous trace map between Sobolev spaces on the manifold and its boundary, cf.\ \cite{AdamsFournier2003, LionsMagenes1972}.
Finally, let us define some particular extensions of the Laplace-Beltrami operator which we
  will use in the following sections.
  We refer to \cite{LionsMagenes1972} for further details.

  \begin{definition}
    \label{def:min_max_laplacian}
    Let $\Delta$ denote the Laplace-Beltrami operator on a Riemannian manifold with boundary and $\Delta_0 \coloneqq \Delta |_{C_c^\infty(\Omega^\circ)}$, its restriction to smooth functions compactly supported on the interior of $\Omega$.
    We define the following extensions:
    \begin{enumerate}[label=\textit{(\roman*)},nosep, leftmargin=*]
      \item The \emph{minimal closed extension}, $\Delta_{\text{min}}$, is the closure of $\Delta_0$.
      \item The \emph{maximal closed extension} $\Delta_{\text{max}}$ is the adjoint of the
        minimal closed extension, whose domain can be characterised by
        $\dom\Delta_{\text{max}} = \{\Phi \in L^2(\Omega) \mid\Delta \Phi \in L^2(\Omega)\}$.
    \end{enumerate}
  \end{definition}
  For the rest of the work we will denote by Greek capital letters the functions on the spaces
  $H^k(\Omega)$ and by the corresponding lower case letter their trace.
  For $k \geq 2$ the trace of the normal derivatives is well defined as well, and in such a case we will denote it by dotted lower case letters.
  That is, for $\Psi \in H^k(\Omega)$, $k \geq 2$, we denote
  \begin{equation*}
    \psi \coloneqq \gamma(\Psi), \qquad \dot{\psi} \coloneqq \gamma(\d\Psi(\nu)),
  \end{equation*}
  where $\nu \in \mathfrak{X}(\Omega)$ is a vector field pointing outwards, and the orientation of the manifold $\Omega$ and its boundary are chosen in such a way that $\mathrm{i}_\nu \d\mu_\eta = \d\mu_{\partial\eta}$.

    A \emph{magnetic potential} is a real-valued 1-form $A \in \Lambda(\Omega)$.
    We consider a deformed differential associated
    with a magnetic potential,
    \begin{equation*}
      \d_A: C^\infty(\Omega) \to \Lambda(\Omega), \quad \Phi \mapsto \d \Phi + i \Phi A,
    \end{equation*}
    where $\d$ is the exterior derivative and $i$ is the imaginary unit.
    Its formal adjoint,
    \begin{equation*}
      \delta_A: \Lambda(\Omega) \to C^{\infty}(\Omega),
    \end{equation*}
    can be defined by the identity
    \begin{equation*}
      \langle \d_A \Phi, \omega \rangle = \langle \Phi, \delta_A \omega \rangle, \quad
      \Phi \in C^\infty_c(\Omega^\circ),\; \omega \in \Lambda(\Omega).
    \end{equation*}

    We can now define the magnetic Laplacian $\Delta_A$ associated with the potential
    $A \in \Lambda(\Omega)$ by the formula
    \begin{equation*}
      \Delta_A = \delta_A \d_A: C^\infty(\Omega) \to C^\infty(\Omega).
    \end{equation*}
    As in the case of the Laplace-Beltrami operator, Stokes' theorem implies that $\d_A^* = \delta_A$ for a
    boundaryless manifold, and therefore the magnetic Laplacian $\Delta_A$ defined above is essentially self-adjoint on such manifolds.

    A straightforward calculation leads to the relation
    \begin{equation}
      \label{eq:magnetic_laplacian_expanded}
      \Delta_A \Phi = \Delta \Phi - 2i (A, \d \Phi) + \bigl(i \delta A + (A,A)\bigr) \Phi,
    \end{equation}
    where $(\alpha, \beta) = \eta^{-1}(\overline{\alpha}, \beta)$, $\alpha, \beta \in \Lambda(\Omega)$, represents the canonical scalar product on the cotangent bundle induced by the Riemannian metric.
    We shall now concentrate on the definition of the boundary conditions that make these operators self-adjoint, therefore giving rise to well defined quantum dynamics.

\subsection{Self-adjoint extensions and sesquilinear forms}
\label{subsec:sa_extensions_quad_forms}

  One of the goals of the research presented in this work is to characterise the controllability of quantum systems using as space of controls the boundary conditions, i.e., the set of self-adjoint extensions of the Laplace-Beltrami operator.
  It is therefore necessary to have a convenient parametrisation of such a set.
   In the one-dimensional case, the set of self-adjoint extensions of the Laplacian is in one-to-one correspondence with
  the unitary operators on the Hilbert space of boundary data (i.e., the trace of the function and
  its normal derivative at the boundary).
  The correspondence is given through the boundary
  equation,
  \begin{equation}
    \label{eq:boundary_equation}
    \varphi - i\dot{\varphi} = U(\varphi + i \dot{\varphi}),
  \end{equation}
  that defines the domain of the self-adjoint extension associated with the unitary $U$ \cite{AsoreyIbortMarmo2005, BruningGeylerPankrashkin2008, Kochubei1975}.

  In the case of a general manifold $\Omega$, the statement above is no longer true.
  Instead, there are several characterisations of the self-adjoint extensions of symmetric differential operators based on the boundary data.
  For instance, G.~Grubb \cite{Grubb1968} gave a complete characterisation in terms of pseudo-differential operators acting on the Sobolev spaces over the boundary.
  In general it is not an easy task to connect this characterisation with the boundary data.
  {
  The theory of boundary triples can be generalised to quasi boundary triples, which apply to this case \cite{BehrndtLanger2012}.
  Another generalisation of the boundary triples formalism which exploits the relationship between positive self-adjoint operators and quadratic forms is that of boundary pairs \cite{Post2016}.
  For a discussion on the relationship between these approaches, see \cite[Sections~1.4 and 1.5]{Post2016}.
  Instead, we are following the approach in \cite{IbortLledoPerezPardo2015} which also bases on positive sesquilinear forms, being in that sense closer to that of boundary pairs.
  This approach enables a straightforward characterisation of the form-domains considered in terms of an implicit equation like Eq.~\eqref{eq:boundary_equation}, which allows us to select easily self-adjoint extensions with constant form domain.
  Moreover, this formalism enables us to find a uniform lower bound for our operators that depends directly on the parameters describing the boundary conditions.
  The importance of this will become clear on later sections since both, the constant form-domain and the uniform lower bound, play a crucial role in the existence of solutions for the Schrödinger Equation and the proof of Theorem~\ref{lemma:magnetic_family_formlinear}.
  }%
  Let us introduce some definitions before summarising the characterisation result.

  \begin{definition}
    A unitary operator $U: L^2(\partial \Omega) \to L^2(\partial \Omega)$ is said to have \emph{gap at $-1$} if
    either $U + \mathbb{I}$ is invertible or $-1$ is in the spectrum of $U$ but it is not an accumulation
    point of $\sigma(U)$.
  \end{definition}

  Consider a unitary $U$ with gap at $-1$.
  Let $P$ denote the orthogonal projector onto the eigenspace associated with $-1$ and $P^\perp = \mathbb{I} - P$.
  The \emph{partial Cayley transform} of $U$ is the linear operator on
  $L^2(\partial \Omega)$ defined by
  \begin{equation*}
    \mathcal{C}_U \varphi \coloneqq i P^\perp \frac{U - \mathbb{I}}{U + \mathbb{I}} \varphi.
  \end{equation*}
  Using the spectral resolution of $U$, it can be shown that $\mathcal{C}_U$ is a bounded, self-adjoint
  operator on $L^2(\partial \Omega)$, cf.\ \cite[Prop.~3.11]{IbortLledoPerezPardo2015}.

  \begin{definition} \label{def:admissible_unitary}
    A unitary $U$ with gap at $-1$ is said to be \emph{admissible} if its partial Cayley transform,
    $\mathcal{C}_U$, leaves invariant the fractional Sobolev space $H^{1/2}(\partial \Omega)$ and it is
    continuous on it, i.e.,
    \begin{equation*}
      \norm{\mathcal{C}_U \varphi}_{H^{1/2}(\partial \Omega)} \leq K \norm{\varphi}_{H^{1/2}(\partial \Omega)},
      \qquad \varphi \in H^{1/2}(\partial \Omega).
    \end{equation*}
  \end{definition}

This implies that the projection $P$ is continuous on $H^{1/2}(\partial \Omega)$ and therefore $P\circ\gamma:H^1\to H^{1/2}(\partial \Omega)$ is also continuous.
  We can now restrict our attention to admissible unitary operators and
  associate to them a sesquilinear form.
  Eventually, this sesquilinear form will allow us to associate a
  self-adjoint extension of the Laplace-Beltrami operator to each of such unitaries.

  \begin{definition}
    \label{def:associated_form}
    Let $U$ be an admissible unitary operator and let $A \in \Lambda(\Omega)$ be a magnetic potential.
    The \emph{sesquilinear form associated with $U$ and $A$} is defined by
    \begin{equation*}
      Q_{A,U}(\Phi, \Psi) = \langle \d_A\Phi, \d_A\Psi \rangle
      - \langle \varphi, \mathcal{C}_U \psi \rangle, \qquad
      \dom Q_{A,U} = \{\Phi \in H^1(\Omega) \mid P \varphi = 0\},
    \end{equation*}
    where $P$ is the projector onto the eigenspace associated with the eigenvalue $-1$ of $U$ and $\varphi = \gamma(\Phi)$, $\psi = \gamma(\Psi)$.
  \end{definition}

  Notice that the unitary $U=-\mathbb{I}$ corresponds to Dirichlet boundary conditions, a fact that will be used later.
  Whenever $A$ is the zero magnetic potential, denoted $\mathcal{O}$, we will shorten the notation and denote by $Q_U := Q_{\mathcal{O}, U}$ the sesquilinear form
  associated with the admissible unitary $U$.
  For such a case, i.e., $A = \mathcal{O}$, it is proven in \cite[Thm.~4.9]{IbortLledoPerezPardo2015} that forms constructed this way are semibounded from below, that is, there is $m>0$ depending only on $\Omega$ and $\|\mathcal{C}_U\|$ such that $Q_{U}(\Phi, \Phi) \geq -m \norm{\Phi}^2$.

  We will show next that this is also the case for non-vanishing magnetic potentials.
  Before we proceed, let us show the following technical lemma.
  \begin{lemma} \label{lemma:dphi_Q_bounded}
    For any unitary $U$ with gap at $-1$, there is a constant $C$ depending only on $\Omega$ and $\|\mathcal{C}_U\|$ such that $C$ is linear on $\|\mathcal{C}_U\|$ and the associated sesquilinear form $Q_U$ satisfies the inequality
    \begin{equation*}
      \|\d\Phi\|^2 \leq 2 Q_U(\Phi,\Phi) + (C + 1) \|\Phi\|^2.
    \end{equation*}
  \end{lemma}
  \begin{proof}
    By the Cauchy-Schwarz inequality and \cite[Theorem 1.5.1.10]{Grisvard1985}, we have
    \begin{equation*}
      |\langle \varphi, \mathcal{C}_U \varphi \rangle| \leq \norm{\mathcal{C}_U} \norm{\varphi}^2
      \leq \frac{1}{2} \norm{\Phi}_1^2 + \frac{C}{2} \norm{\Phi}^2,
    \end{equation*}
    for some constant $C$.
    
    Substituting into the definition of $Q_U$ we get that
    \begin{equation*}
      Q_U(\Phi, \Phi) = \norm{\d\Phi}^2 - \langle \varphi, \mathcal{C}_U \varphi \rangle
      \geq \frac{1}{2} \norm{\d\Phi}^2 - \frac{1}{2}(C + 1) \norm{\Phi}^2,
    \end{equation*}
    from which the result follows.
  \end{proof}

  \begin{proposition}\label{prop:graph_norm_magnetic}
    For any magnetic potential $A \in \Lambda(\Omega)$ and any admissible unitary operator $U$ with gap at $-1$ the following holds:
    \begin{enumerate}[label=\textit{(\roman*)},nosep, leftmargin=*]
      \item \label{prop:graph_norm_magnetic_1}
        $Q_{A,U}$ is semibounded from below.
      \item \label{prop:graph_norm_magnetic_2}
        The graph norm $\norm{\argdot}_{Q_{A,U}}$ is equivalent to $\norm{\argdot}_{Q_U}$.
    \end{enumerate}
  \end{proposition}

  \begin{proof}
    By definition one has
    \begin{equation*}
      \langle \d_A \Phi, \d_A \Phi \rangle = \|\d\Phi\|^2 + \|\Phi A\|^2 + 2 \Im \langle \d\Phi, \Phi A \rangle.
    \end{equation*}
    Now, 
    \begin{alignat*}{2}
      Q_{A,U}(\Phi, \Phi) &= \langle \d_A \Phi, \d_A \Phi \rangle - \langle \varphi, \mathcal{C}_U \varphi \rangle \\
      &\geq Q_U(\Phi, \Phi) + \|\Phi A\|^2 - 2 |\langle \d\Phi, \Phi A \rangle|.
    \end{alignat*}
    By the Cauchy-Schwarz inequality and Young inequality with $\varepsilon$ we have
    \begin{equation*}
      Q_{A,U}(\Phi, \Phi) \geq Q_U(\Phi, \Phi) + (1 - \varepsilon) \|\Phi A\|^2 - \frac{1}{\varepsilon}\|\d\Phi\|^2
    \end{equation*}
    for every $\varepsilon>0$.
    Taking $\varepsilon=4$ and applying Lemma~\ref{lemma:dphi_Q_bounded} it follows
    \begin{equation*}
      Q_{A,U}(\Phi,\Phi) \geq \frac{1}{2} Q_U(\Phi,\Phi) - \left(3 \|A\|^2 + \frac{C+1}{4}\right) \|\Phi\|^2.
    \end{equation*}
    The lower semibound of $Q_{A,U}$ follows now from that of $Q_U$.
    
    Let us now show the equivalence \ref{prop:graph_norm_magnetic_2}.
    On the one hand, from the previous inequality it follows immediately that for some $K>0$
    \begin{equation*}
      \norm{\Phi}_{Q_U} \leq K\norm{\Phi}_{Q_{A,U}}.
    \end{equation*}
    On the other hand, let $m>0$ be the lower bound of $Q_{A,U}$ and $\alpha := \sqrt{\sup_{\Omega} \eta^{-1}(A,A)}$.
    Then,
    \begin{equation*}
      \begin{alignedat}{2}
        \norm{\Phi}_{Q_{A,U}}^2
        &= (m + 1) \norm{\Phi}^2 + \|\d\Phi\|^2 + \|\Phi A\|^2 + 2 \Im \langle \d\Phi, \Phi A \rangle - \langle \varphi, \mathcal{C}_U \varphi \rangle\\
        &\leq (m+\alpha^2+1) \norm{\Phi}^2 + \norm{\d\Phi}^2 + 2 \alpha \|\d\Phi\| \|\Phi\| - \langle\varphi, \mathcal{C}_U \varphi\rangle\\
        &= (m+\alpha^2+1) \norm{\Phi}^2 + Q_U(\Phi, \Phi) + 2 \alpha \|\d\Phi\|\|\Phi\| \\
        &\leq K \|\Phi\|_{Q_U}^2
      \end{alignedat}
    \end{equation*}
    where $K = \max\{m + \alpha^2 + 1, 2\alpha\}$, since all the addends are proportionally smaller than $\|\Phi\|_{Q_U}^2$.
  \end{proof}

  \begin{theorem} \label{thm:quadratic_form_H1bounded}
    For any magnetic potential $A$, the sesquilinear form associated with an admissible unitary with
    gap at $-1$ is closed, i.e., $\dom Q_{A,U}$ is closed with respect to the graph norm,
    $\norm{\argdot}_{Q_{A,U}}$.
  \end{theorem}
  \begin{proof}
    The admissibility condition guarantees that $\dom Q_{A,U}$ is closed in $H^1(\Omega)$, thus it suffices to show that $\norm{\argdot}_{Q_{A,U}}$ is equivalent to $\norm{\argdot}_1$.
    By Proposition~\ref{prop:graph_norm_magnetic} it is enough to show the equivalence for $A$ equal to the zero magnetic potential.
    
    First we have the following inequalities:
    \begin{equation*}
      |\langle \varphi, \mathcal{C}_U \varphi \rangle| \leq \norm{\varphi} \norm{\mathcal{C}_U \varphi}
      \leq \norm{\mathcal{C}_U} \norm{\varphi}^2
      \leq C \norm{\mathcal{C}_U} \norm{\varphi}_{\frac{1}{2}}^2
      \leq C' \norm{\mathcal{C}_U} \norm{\Phi}^2_1,
    \end{equation*}
    where we have used the Sobolev inclusions and the continuity of the trace map.

    Let $m$ be the lower bound of $Q_U$.
    Substituting the previous inequality into the definition of the graph norm, one gets
    \begin{equation*}
      \begin{alignedat}{2}
        \norm{\Phi}_{Q_U}^2 &= \norm{\d \Phi}^2 + (1 + m) \norm{\Phi}^2 -
        \langle \varphi, \mathcal{C}_U \varphi \rangle \\
        &\leq \norm{\d \Phi}^2 + (1 + m) \norm{\Phi}^2 + |\langle \varphi, \mathcal{C}_U \varphi \rangle| \\
        &\leq K^2 \norm{\Phi}_1^2
      \end{alignedat}
    \end{equation*}

    Let us prove now the reverse inequality.
    By Lemma~\ref{lemma:dphi_Q_bounded} we have
    \begin{equation*}
      Q_U(\Phi, \Phi) \geq \frac{1}{2} \norm{\d\Phi}^2 - \frac{1}{2}(C + 1) \norm{\Phi}^2,
    \end{equation*}
from which it follows
    \begin{equation*}
      Q_U(\Phi, \Phi) + \frac{1}{2}(C + 2) \norm{\Phi}^2 \geq \frac{1}{2} \norm{\Phi}_1^2,
    \end{equation*}
which completes the proof.
  \end{proof}

  The previous results lead to the following corollary:
  \begin{corollary} \label{corol:uniform_semibound_equivalence_formnorms}
    Let $\mathcal{A}$ be a family of magnetic potentials and $\mathcal{U}$ a family of admissible unitaries such that the associated forms have domain $\mathcal{H}^+$ and
    \begin{equation*}
      \sup_{U \in \mathcal{U}} \|\mathcal{C}_U\| < \infty
      \quad\text{and}\quad
      \sup_{A \in \mathcal{A}} \|A\|_{\infty} < \infty.
    \end{equation*}
    Then:
    \begin{enumerate}[label=\textit{(\roman*)},nosep,leftmargin=*]
      \item There exists $m>0$ such that for every $A \in \mathcal{A}$ and every $U \in \mathcal{U}$ we have $Q_{A,U}(\Phi, \Phi) \geq -m \|\Phi\|^2$ for every $\Phi\in\H^+$.
      \item Let $A_0\in\mathcal{A}$ and $U_0\in\mathcal{U}$.
        There is a constant $K>0$ such that for every $A \in \mathcal{A}$ and $U \in \mathcal{U}$ we have
        \begin{equation*}
          K^{-1} \|\Phi\|_{Q_{A,U}} \leq \|\Phi\|_{Q_{A_0,U_0}} \leq K \|\Phi\|_{Q_{A,U}}, \qquad \forall\Phi \in \mathcal{H}^+ 
        \end{equation*}
        for every $U \in \mathcal{U}$ and every $A \in \mathcal{A}$.
    \end{enumerate}
  \end{corollary}

  We have established that, for any admissible unitary operator and a magnetic potential satisfying the conditions above, the sesquilinear form associated with them is closed and semibounded from below.
  We will apply the representation theorem, cf.\ \cite[Sec.~VI.2]{Kato1995}, to define the self-adjoint magnetic Laplacian.

  \begin{definition}\label{def:magnetic-laplacian}
    Let $\Omega$ be a compact, Riemannian manifold with boundary and let $U\colon L^2(\partial\Omega)\to L^2(\partial\Omega)$ be an admissible unitary operator.
    Let $A\in\Lambda(\Omega)$.
    The \emph{magnetic Laplacian} operator associated with the unitary $U$ and the magnetic potential $A$, denoted by $\Delta_{A,U}$, with domain $\dom \Delta_{A,U}$, is the unique self-adjoint operator associated with the closed, semibounded sesquilinear form $Q_{A,U}$.
    In the case in which the magnetic potential $A$ is identically zero, we will refer to this operator as the \emph{Laplace-Beltrami} operator and denote it and its domain respectively by $\Delta_U$ and $\dom \Delta_U$.
  \end{definition}
  
  For the Laplace-Beltrami operator associated with the unitary $U$, one can provide the following characterisation of its domain.
  \begin{theorem}
    \label{thm:sa_extensions_laplacian}
    Let $\Delta_U$, densely defined on $\dom \Delta_U$, be the Laplace-Beltrami operator associated with the unitary $U$.
    Then 
    \begin{equation*}
      H^2(\Omega) \cap \dom \Delta_U = \{\Phi \in H^2(\Omega) \mid \varphi - i\dot{\varphi} = U(\varphi + i \dot{\varphi})\}.
    \end{equation*}
  \end{theorem}
  \begin{proof}
    Let $\Phi \in H^2(\Omega) \cap \dom \Delta_U$ and $\Psi\in\dom{Q_U}$.
    Integrating by parts in the definition of $Q_U$, cf.\ Definition~\ref{def:associated_form}, yields 
    \begin{equation*}
      Q_U(\Psi, \Phi) = \langle \d\Psi, \d\Phi \rangle - \langle \psi, \mathcal{C}_U\varphi \rangle
      = \langle \Psi, \Delta\Phi \rangle + \langle \psi, \dot{\varphi} - \mathcal{C}_U\varphi \rangle.
    \end{equation*}
    From this it follows that
    \begin{equation}
      \label{eq:Delta_U_dom}
      \langle \Psi, \Delta_U \Phi - \Delta \Phi \rangle
      = \langle \psi, \dot{\varphi} - \mathcal{C}_U \varphi \rangle.
    \end{equation}
    Thus, for every $\Psi \in H^1_0(\Omega) \cap H^2(\Omega) \subset \mathcal{D}(Q_U)$, we have
    \begin{equation*}
      \langle \Psi, \Delta_U \Phi - \Delta \Phi \rangle = 0,
    \end{equation*}
    which implies $\Delta_U \Phi = \Delta \Phi$ since $H^1_0(\Omega) \cap H^2(\Omega)$ is dense in
    $L^2(\Omega)$.
    Substituting this in Eq.~\eqref{eq:Delta_U_dom}, it follows
    \begin{equation*}
      \langle \psi, \dot{\varphi} - \mathcal{C}_U \varphi \rangle = 0.
    \end{equation*}
    Since $\Psi \in \mathcal{D}(Q_U)$, $\psi = P^{\bot}\psi$ and the equation above is equivalent to
    \begin{equation*}
      P^\perp \dot{\varphi} = \mathcal{C}_U \varphi.
    \end{equation*}
    Additionally, since $\Phi \in \dom \Delta_U\subset \dom Q_U$, we have $P\varphi = 0$.
       
    Applying the orthogonal projector onto $\ker(U + 1)$, $P$, to the equation $\varphi - i\dot{\varphi} = U(\varphi + i \dot{\varphi})$
    it follows
    \begin{equation*}
      P\varphi - iP\dot{\varphi} = PU (\varphi + i \dot{\varphi}) = -P\varphi - i P\dot{\varphi},
    \end{equation*}
    which is equivalent to $P\varphi = 0$.
    On the other hand, projecting with $P^\perp = \mathbb{I} - P$,
    one gets
    \begin{equation*}
      P^\perp\dot{\varphi} = iP^\perp \frac{U - \mathbb{I}}{U + \mathbb{I}} \varphi = \mathcal{C}_U \varphi.
    \end{equation*}
  \end{proof}

  The previous result motivates the following definition.

  \begin{definition}\label{def:boundary_equation_operator}
    Let $U\colon L^2(\partial\Omega) \to L^2(\partial\Omega)$ be an admissible unitary operator and $A\in\Lambda(\Omega)$ a magnetic potential.
    The \emph{domain associated with the unitary $U$ and the magnetic potential $A$} is the set
    $$\D_{U,A} = \left\{ \Phi\in\H^2(\Omega) \mid \varphi - i\dot{\varphi}_A = U(\varphi + i \dot{\varphi}_A) \right\},$$
    where $\dot{\varphi}_A:=\gamma\left(\mathrm{d}_A\Phi(\nu)\right)$, and $\nu$ is the normal vector field to the boundary $\partial \Omega$ pointing outwards.
    In the case in which the magnetic potential is identically zero, we will drop the subindex and denote it by $\D_U$.
  \end{definition}

  \begin{proposition} \label{prop:magnetic_form_equivalence}
    Let $A \in \Lambda(\Omega)$ be an exact, magnetic potential, i.e.\ $A = \d\Theta$ for some function $\Theta \in C^2(\Omega)$, and denote $\theta = \gamma(\Theta)$.
    Define the unitary operator $J: \Phi \in L^2(\Omega) \mapsto e^{i\Theta}\Phi \in L^2(\Omega)$.
    Then $J(\dom Q_{A,U}) = \dom Q_{e^{i\theta}Ue^{-i\theta}}$ and
    \begin{equation*}
      Q_{A,U}(\Psi, \Phi) = Q_{e^{i\theta} U e^{-i\theta}}(J\Psi, J\Phi).
    \end{equation*}
  \end{proposition}
  \begin{proof}
    Let us first show that $J(\dom Q_{A,U}) = \dom Q_{e^{i\theta}Ue^{-i\theta}}$.
    Since $\Theta \in C^2(\Omega)$, for any $\Phi \in H^1(\Omega)$ we have $J\Phi \in H^1(\Omega)$ and $\gamma(J\Phi) = e^{i\theta} \varphi$.
    Denote $\tilde{U} = e^{i\theta}Ue^{-i\theta}$ and let $\tilde{P}$ be the orthogonal projector onto $\ker(\tilde{U}+1)$.
    It is straightforward to check that $\tilde{U} \gamma(J\Phi) = - \gamma(J\Phi)$ if and only if $\Phi \in \ker(U+1)$, and therefore $\tilde{P} = e^{i\theta} P e^{-i\theta}$ where $P$ is the orthogonal projector onto $\ker(U+1)$.
    A straightforward calculation shows the following relation between the partial Cayley transforms:
    \begin{equation*}
      \mathcal{C}_{\tilde{U}} = e^{i\theta} \mathcal{C}_U e^{-i\theta}.
    \end{equation*}
    From the above conditions it follows that $\tilde{U}$ is an admissible unitary operator and that $\Phi \in \dom Q_{A,U}$ if and only if $J\Phi \in \dom Q_{\tilde{U}}$.

    Finally, $\d(J\Phi) = J\d_A \Phi$, and the following identity holds:
    \begin{equation*}
      Q_{A,U}(\Psi,\Phi) = \langle J\d_A\Psi, J\d_A\Phi \rangle - \langle \psi, e^{-i\theta}\mathcal{C}_{\tilde{U}}e^{i\theta}\varphi \rangle
      = Q_{\tilde{U}}(J\Psi, J\Phi),
    \end{equation*}
    where we have used the unitarity of $J$.
  \end{proof}
  
   As a corollary of Theorem~\ref{thm:sa_extensions_laplacian}, one can show an equivalent result for magnetic Laplacians.
 
  \begin{corollary}\label{corol:magnetic_laplacian_equivalence}
    Let $A \in \Lambda(\Omega)$ be an exact magnetic potential, i.e.\ $A = \d\Theta$ for some function $\Theta \in C^2(\Omega)$, $U: L^2(\partial\Omega)\to L^2(\partial\Omega)$ an admissible unitary operator, $J$ the unitary operator $J: \Phi \in L^2(\Omega) \mapsto e^{i\Theta} \Phi \in L^2(\Omega)$ and $\theta = \gamma(\Theta)$.
    Let $\Delta_{A,U}$ be the magnetic Laplacian, densely defined on $\dom \Delta_{A,U}$, associated with $U$ and $A$.
    Then
    \begin{equation*}
      H^2(\Omega) \cap \dom \Delta_{A,U} = \mathcal{D}_{A,U}.
    \end{equation*}
    Moreover, $\dom {J(}\Delta_{A,U}{)} = \dom \Delta_{e^{i\theta}Ue^{-i\theta}}$ and $\Delta_{A, U} = J^{-1}\Delta_{e^{i\theta} U e^{-i\theta}}J$, where $\Delta_{e^{i\theta}Ue^{-i\theta}}$ is the Laplace-Beltrami operator associated with the unitary $e^{i\theta} U e^{-i\theta}$.
  \end{corollary}
  \begin{proof}
    Denote $\tilde{U} = e^{i\theta} U e^{-i\theta}$.
    From Proposition~\ref{prop:magnetic_form_equivalence} we have $J(\dom Q_{A,U}) = \dom Q_{e^{i\theta}Ue^{-i\theta}}$ and $Q_{A,U}(\Psi, \Phi) = Q_{\tilde{U}}(J\Psi, J\Phi)$.
    Therefore, cf.\ \cite[Sec VI.2]{Kato1995}, $\Phi \in \dom \Delta_{A,U}$ if and only if $J\Phi \in \dom \Delta_{\tilde{U}}$ and it follows that $\Delta_{A, U} = J^{-1}\Delta_{\tilde{U}}J.$

    Assume that $\Phi\in \H^2(\Omega) \cap \dom \Delta_{A, U}$.
    Then $J\Phi \in \H^2(\Omega) \cap \dom \Delta_{\tilde{U}} = \D_{\tilde{U}}$, which implies
    $$\gamma(J\Phi) - i\gamma (\mathrm{d}(J\Phi)(\nu)) = \tilde{U} (\gamma(J\Phi) + i\gamma (\mathrm{d}(J\Phi)(\nu))),$$
    for $\nu$ the normal vector field to the boundary.
    Now noticing that {$\mathrm{d}(J\Phi)(\beta) = J\mathrm{d}_A\Phi(\beta)$}, $\beta\in\mathfrak{X}(\Omega)$, we get 
$e^{i\theta}\varphi-i e^{i\theta}\dot{\varphi}_A = \tilde{U}(e^{i\theta}\varphi+i e^{i\theta}\dot{\varphi}_A).$
The converse inclusion is proven in a similar way.
  \end{proof}
  
%%%%%%%%%%%%%%%%%%%%%%%%%%%%%%%%%%%%%%%%%%%%%%%%%%%%%%%%%%%%%%%%%%%%%%%%%%%%%%%%%%%%%%%%%%%%%%%%%%%%%%%%%%%%%%%%%%%%%%%%%%%%%%%%%%%%%%%%%%%%%%%%%%%%%%%%%%%%

\section{Thick Quantum Graphs and quasi-$\delta$ boundary conditions}\label{sec:TQG}

  A large class of self-adjoint extensions of the Laplace-Beltrami operator on a Riemannian manifold with smooth boundary can be described with the formalism depicted in the previous section.
  As already stated in the introduction, our main motivation is to study the problem of Quantum Control at the boundary in circuit-like settings.
  We are interested in presenting results in higher-dimensional analogues of Quantum Graphs.
  There are several ways for implementing such generalisation.
  For instance, there are works describing higher-dimensional graph-like spaces generalising Quantum Graphs, which are particularly useful when studying the limit on which this higher-dimensional spaces converge to a graph (see \cite{Post2012} and references therein).
  However, we are going to take an approach focused instead on the boundary conditions, which contains the situation of Quantum Graphs as particular cases.

  Thick Quantum Graphs are meant to describe the dynamics of a quantum particle on a closed circuit, made out of a collection of wires interconnected through some {interfaces.
  We shall give now the precise definition, whose components are illustrated in Figure~\ref{fig:example-notation}.
  At the end of this subsection there are some meaningful examples of this construction; see also Fig.~\ref{fig:examples}.

  \begin{figure}[t]
    \renewcommand{\thesubfigure}{\Alph{subfigure}}
    \begin{subfigure}[b]{0.45\textwidth}
      \centering
      \includegraphics[width=\textwidth]{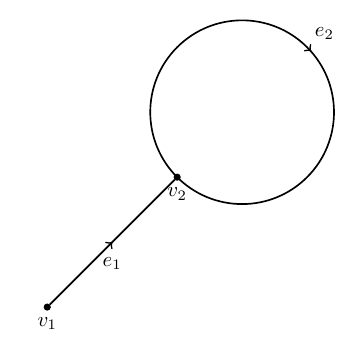}
      \caption{}\label{fig:example-notation-graph}
    \end{subfigure}
    \begin{subfigure}[b]{0.45\textwidth}
      \centering
      \includegraphics[width=\textwidth]{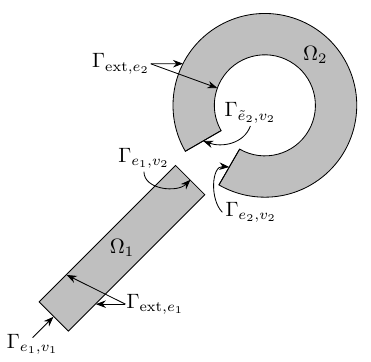}
      \caption{}\label{fig:example-notation-bands}
    \end{subfigure}
    \caption{
      {
        Illustration of the elements in Definition~\ref{def:thick-quantum-graph} for a two-dimensional Thick Quantum Graph with 2D bands as wires.
        The graph (A) describes the structure of the connections while (B) represents both the wire manifolds and the partition of their boundaries.
        Note the notation $e_2, \tilde{e}_2$ on the part of the boundary of $\Omega_2$ participating of the connection in vertex $v_2$, accounting for the two ends of the loop participating on the connection at that vertex, c.f.\ Remark~\ref{remark:loops-edges}.
      }
    }
    \label{fig:example-notation}
  \end{figure}  
  }

  \begin{definition}\label{def:thick-quantum-graph}
    An \emph{$n$-dimensional Thick Quantum Graph} is a triple $(G, \Omega, \Gamma)$ where:
    \begin{enumerate}[label=\textit{(\roman*)},nosep, leftmargin=*]
      \item $G = (V, E)$ is a finite graph with vertex set $V$ and edge set $E$.
      \item $\Omega$ is a non-connected $n$-dimensional Riemannian manifold with as many connected components as the edges of the graph, i.e., $\Omega=\cup_{e\in E}\Omega_e$ where each $\Omega_e$ is a connected Riemannian manifold with non-trivial boundary.
      \item $\Gamma$ is a partition of the boundary $\partial\Omega$ defined as follows.
        Let $V_e\subset V$ denote the set of vertices joined by the edge $e\in E$.
        Let $\Gamma_{\mathrm{ext},e}$ and $\Gamma_{v,e}$, with $e \in E$ and $v\in V_e$, be connected submanifolds of $\partial \Omega_e$, the boundary of the oriented manifold $\Omega_e$.
        We will require that for any $\Gamma_{v,e}$, $e\in E$, $v\in V_e$, there exists a connected open neighbourhood $\mathcal{V}_{v,e}\subset\partial\Omega$, with $\overline{\Gamma_{v,e}}\subset\mathcal{V}_{v,e}$, 
        such that the latter are pairwise disjoint, i.e.\ $\mathcal{V}_{v,e}\cap \mathcal{V}_{v',e} = \emptyset$ for $e\in E$ and $v,v'\in V_e$ with $v\neq v'$.
        Then $$ \partial \Omega= \bigsqcup_{e\in E} \Gamma_{\mathrm{ext},e}\bigsqcup_{v\in V_e} \Gamma_{v,e}.$$
          \end{enumerate}
  \end{definition} 
  {
  \begin{remark}\label{remark:loops-edges}
    As is it customary when dealing with loops, if an edge forms a loop it must appear twice in $E_v$, allowing to account for the two endpoints of the edge that take part in the vertex, cf.\ \cite[Definition 2.1.1]{Post2012}.
    For concrete examples of this, see Figure~\ref{fig:example-notation} and Examples~\ref{ex:5} and \ref{ex:6}.
  \end{remark}
}
    
  When it leads to no confusion, we will abuse the notation and denote by $\Omega$ the Thick Quantum Graph itself.
  Also, for the sake of simplicity we will assume from now on that every $\Omega_e$ of a Thick Quantum Graph is a compact differentiable manifold with piecewise smooth and compact boundary.

  {
  Let us add some comments on the role of its components.
  A Thick Quantum Graph is made up by a graph $G$, a possibly non-connected Riemannian manifold with boundary $\Omega$, and a partition $\Gamma$ of $\partial\Omega$.
  The graph $G$ represents the topology of the physical circuit, each of its edges representing a \emph{wire} and each of its vertices representing an \emph{interface}.
  The role of the manifold $\Omega$ is to describe the shape of the \emph{wires} and so each of its connected components describes one of the \emph{wires}.
  Although it could be possible to consider situations in which the connected components of the manifold $\Omega$ have different dimension, cf.\ \cite{CarlonePosilicano2017,ExnerSeba1987}, we restrict ourselves to the case in which all the components have the same dimension and thus $\Omega$ is a differentiable manifold of given dimension.
  Finally, the role of the partition $\Gamma$ is twofold.
  On the one hand, it specifies which parts of the boundary of each \emph{wire} is external, and thus not participating on any interface.
  On the other hand, it specifies the structure of the \emph{interfaces}, describing which part of the boundary of each \emph{wire} participates on which \emph{interface}.
  In the definition, there is also a technical condition ensuring that, for each edge, the different pieces of the boundary participating on interfaces are separated one from each other, as it will be important when proving Theorem~\ref{thm:quasi_delta_BCS_equivalence}.
}

  {
  Before we move into the discussion of the boundary conditions for our systems, let us add a few words on the relation between our notion of Thick Quantum Graphs and some other higher-dimensional versions of quantum graphs.
  Generally, other versions are built as \emph{fattened} versions of quantum graphs, the interest put on the limit on which its width converges to $0$ and the system becomes a quantum graph.
  However, we are not interested on the thin limit as a way of modelling a thin structure by a graph with certain conditions.
  We aim at describing a non-one-dimensional situation with interfaces.
  These interfaces are considered infinitesimally thin and modelled by boundary conditions (see Figure~\ref{fig:TQG-vs-fattened})

  \begin{figure}[b]
    \renewcommand{\thesubfigure}{\Alph{subfigure}}
    \begin{subfigure}[b]{0.45\textwidth}
      \centering
      \includegraphics[scale=1.25]{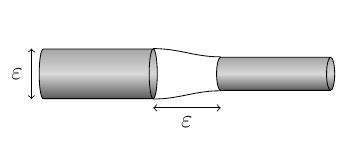}
      \caption{}\label{subfig:comparative-fattened}
    \end{subfigure}
    \begin{subfigure}[b]{0.45\textwidth}
      \centering
      \includegraphics[scale=1.25]{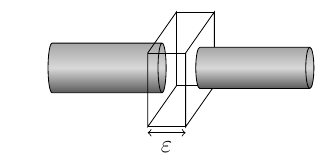}
      \caption{}\label{subfig:comparative-TQG}
    \end{subfigure}
    \caption{
      {
        On the left, a neighbourhood-like fattened generalisation of a Quantum Graph.
        On that case, the thickness of the \emph{edges} is sent to 0 at the same time as the size of the vertex neighbourhoods (i.e., the part of the fat graph \emph{connecting} the two \emph{edges}) \cite{Post2012}.
        On the right, an example of our notion of Thick Quantum Graphs: two wires with different thickness are connected through an interface.
        When one considers the thickness of the interface to be very small, boundary conditions appear as an effective description.
      }
    }
    \label{fig:TQG-vs-fattened}
  \end{figure}
  }
  
  We will only consider \emph{finite} Thick Quantum Graphs, i.e.\ having finitely many wires.
  In this case, the relevant Hilbert spaces defined for the complete graph are just the direct sum of the corresponding Hilbert spaces for its wires.
  For instance,
  \begin{equation*}
    L^2(\Omega) = \bigoplus_{e \in E} L^2(\Omega_e);\qquad
    \langle \Phi, \Psi \rangle_{L^2(\Omega)} = \sum_{e \in E} \langle \Phi_e, \Psi_e \rangle_{
    L^2(\Omega_e)},
  \end{equation*}
  and similarly for the Sobolev spaces.
  
  It is also worth to make some comments about the Hilbert space at the boundary.
  It is clear that, even in this case, one has $L^2(\partial\Omega) = \bigoplus_{e \in E} L^2(\partial\Omega_e)$; however, for the boundary space, a further decomposition will be needed.
  Let $v\in V$ and let $E_v\subset E$ be the subset of edges that share the vertex $v$.
  Denote by $\Gamma_v = \bigcup_{e \in E_v} \Gamma_{v,e}$ the part of the boundary of $\Omega$ that is involved in the connection represented by the vertex $v \in V$, and by $\Gamma_{\text{ext}} = \bigcup_{e \in E} \Gamma_{\text{ext},e}$ the subset of the boundary of $\Omega$ that is not part of a junction; i.e., the external boundary of the Thick Quantum Graph.
  In the one-dimensional case, the latter is formed by the exterior vertices of the graph.
  Therefore, we can write the following decomposition:
  \begin{equation*}
    L^2(\partial\Omega) = L^2(\Gamma_{\text{ext}}) \oplus \bigoplus_{v \in V} L^2(\Gamma_v).
  \end{equation*}
  Particles inside the Thick Quantum Graph need to be able to move
  \emph{freely} throughout the \emph{wires} and to \emph{jump} from a piece of wire to an adjacent
  one: therefore, we are going to consider the quantum evolution given by the free Hamiltonian,
  that is, the Laplacian $\Delta$ which acts wire by wire, i.e.\, for any $\Phi = \bigoplus_{e \in
  E} \Phi_e$ in $H^2(\Omega)$, $$\Delta\Phi = \bigoplus_{e \in E} \Delta_e \Phi_e,$$ with
  $\Delta_e$ the Laplace-Beltrami operator in $\Omega_e$.
  Since we are considering the Thick Quantum Graph $\Omega$ as the collection of the different wires $\Omega_e$, we will model the transitions from one $\Omega_e$ to another by imposing some particular boundary conditions involving the parts of the boundary which are (physically) connected.
  It is natural to assume that allowing the particles to move through this boundary regions should result in a continuous probability of finding the particle at the junctions $v \in V$: hence, the boundary conditions imposed should implement this requirement.
  Furthermore, particles are not supposed to escape the graph through $\Gamma_{\text{ext}}$ which makes natural to impose Dirichlet boundary conditions in $\Gamma_{\text{ext}}$.
  Other types of local boundary conditions on $\Gamma_{\text{ext}}$, such as Neumann or Robin, could be considered as well; the choice of these other types of boundary conditions would not alter significantly the results presented in what follows.

  From this discussion, it is clear that the unitaries implementing (see Subsection \ref{subsec:sa_extensions_quad_forms}) the physically relevant boundary conditions can be written as
  \begin{equation} \label{eq:total_unitary_structure}
    U = U_D \oplus \bigoplus_{v \in V} U_v,
  \end{equation}
  where $U_D: L^2(\Gamma_{\text{ext}}) \to L^2(\Gamma_{\text{ext}})$ is given by $U_D = -\mathbb{I}_{L^2(\Gamma_{\text{ext}})}$ and $U_v: L^2(\Gamma_v) \to L^2(\Gamma_v)$ are the unitaries specifying the boundary conditions in distinct regions of the boundary of $\Omega$.
  
      Some practical information can be extracted from the proof of Theorem~\ref{thm:sa_extensions_laplacian}.
    It has been established that the equation defining $\mathcal{D}_U$, which we will call \emph{boundary equation} since it fixes conditions on the boundary data of $H^2(\Omega)$ functions in the domain, is equivalent to
    \begin{equation*}
      P \varphi = 0, \qquad P^\perp \dot{\varphi} = \mathcal{C}_U \varphi.
    \end{equation*}
    A consequence of Theorem~\ref{thm:sa_extensions_laplacian} is that, for functions in $H^2(\Omega)$, boundary conditions involving only the trace of the functions in the domain must be implemented through the equation $P\varphi=0$.
    Therefore, these conditions should fix the eigenspace associated with the eigenvalue $-1$ of the unitary $U$.
    
  A particular family of self-adjoint extensions preserving the topology of a Thick Quantum Graph,
  in the sense described previously, is what we call quasi-$\delta$ boundary conditions, which are a generalisation of the standard periodic boundary conditions (Kirchhoff vertex conditions in the case of Quantum Graphs).
  In the simplest case, i.e.\ a vertex connecting two edges, Kirchhoff boundary conditions consist on identifying the two joined boundaries, by requiring the functions to be continuous at that vertex and the conservation of the so-called \emph{flux of normal derivatives}.
  The term quasi-Kirchhoff is used when allowing for a relative phase-change.
  Quasi-$\delta$ is a further generalisation allowing the derivatives of the function to be discontinuous at the vertex.
  We will be more specific in the following.

  In order to keep the mathematical description as simple as possible, from now on we are going to focus on a particular way of connecting the wires.
  Let $v\in V$ and denote by $E_v\subset E$ the subset of edges that share the vertex $v$.
  If an edge forms a loop it must appear twice in $E_v$, allowing to account for the two endpoints of $\partial\Omega_e$ to take part in the junction, cf.\ \cite[Definition 2.1.1]{Post2012}.
  We will assume that, for every $e\in E_v$, {all $\Gamma_{v,e}$ are diffeomorphic to each other, and we will chose one of them as reference and denote it as ${\Gamma_{v,0}}$. Thus, there is a diffeomorphism $g_{v,e}: {\Gamma_{v,0}} \to \Gamma_{v,e}$.}
  Let $\d\mu_0$ be a fixed volume element in ${\Gamma_{v,0}}$.
  Using the diffeomorphism we can define an isometry $T_{v,e}: L^2(\Gamma_{v,e}) \to L^2({\Gamma_{v,0}}, \d\mu_0)$ by
  \begin{equation} \label{eq:definition_Te}
    T_{v,e} \varphi \coloneqq \sqrt{|J_{v,e}|} (\varphi \circ g_{v,e}),
  \end{equation}
  where 
  $|J_{v,e}|$ is the Jacobian determinant of the transformation $g_{v,e}$, i.e.\ the proportionality factor between the pull-back of the induced Riemannian metric at the boundary and the reference volume at ${\Gamma_{v,0}}$, i.e., $(g_{v,e})^* \d\mu_{\Gamma_{v,e}} = |J_{v,e}|\d\mu_0$.
  Using these isometries we can define the unitaries $\mathbb{T}_v: L^2(\Gamma_v) \to \bigoplus_{e \in E_v} L^2({\Gamma_{v,0}})$ as the direct sum $\mathbb{T}_v = \bigoplus_{e \in E_v} T_{v,e}$.
  In the case of quasi-Kirchhoff (and also quasi-$\delta$) boundary conditions, the structure of the unitaries $U_v$ in \eqref{eq:total_unitary_structure} can be easily written using this notation.

  \begin{definition}
    \label{def:quasi_delta_BC}
    For every $v \in V$, and every $e \in E_v$, let $\chi_{v,e}\in [0,2\pi),\, \delta_{v} \in (-\pi,\pi)$.
    We call \emph{quasi-$\delta$ boundary conditions
    with parameters} $\chi_{v,e},\, \delta_{v}$ the {boundary conditions} given by the unitary (see Equation \eqref{eq:total_unitary_structure}),
    \begin{equation*}
      U = U_D \oplus \bigoplus_{v \in V} U_v, \quad\text{with}\quad
      U_v = \mathbb{T}_v^*\left( (e^{i \delta_{v}} + 1) P_{v}^\perp - \mathbb{I} \right) \mathbb{T}_v,
    \end{equation*}
    where $P_v^\perp: \bigoplus_{e \in E_v} L^2({\Gamma_{v,0}}) \to \bigoplus_{e \in E_v} L^2({\Gamma_{v,0}})$ is given by the blocks
    \begin{equation*}
      (P_{v}^\perp)_{ee'} = \frac{1}{|E_v|} e^{i (\chi_{v,e} - \chi_{v,e'})},\quad e,e'\in E_v.
    \end{equation*}
  \end{definition}

  By a straightforward calculation, it can be checked that $P_v^\bot$ is an orthogonal projection on $L^2(\Gamma_v)$ and that $P_v:=\mathbb{I}-P_v^\bot$ satisfies $\mathbb{T}^*_vP_v\mathbb{T}_vU_v = U_v\mathbb{T}_v^* P_v \mathbb{T}_v=-\mathbb{T}^*_vP_v\mathbb{T}_v$ and $\mathbb{T}^*_vP^\bot_v\mathbb{T}_vU_v = U_v\mathbb{T}_v^* P_v^\bot \mathbb{T}_v = e^{i\delta_v}\mathbb{T}_v^* P_v^\bot \mathbb{T}_v$ which shows that $U_v$ has eigenvalues $-1$ and $e^{i\delta_v}$.
  Consequently, 
  \begin{equation*}
    P = \mathbb{I}_{L^2(\Gamma_{\text{ext}})} \oplus \bigoplus_{v \in V} \mathbb{T}_v^* (1 - P_v^\perp) \mathbb{T}_v,
  \end{equation*}
  is the orthogonal projector onto the eigenspace of $U$ associated with $-1$, and 
  \begin{equation*}
    P^\perp = 1-P = 0_{L^2(\Gamma_{\text{ext}})} \oplus \bigoplus_{v \in V} \mathbb{T}_v^* P_{v}^\perp \mathbb{T}_v.
  \end{equation*}
  
  {
  \begin{remark}\label{remark:constant-delta-parameters}
    On this definition, we assume that $\chi_{v,e}$ and $\delta_v$ are constant parameters.
    This is a convenient, but not essential, simplification from the more general case on which $\chi_{v,e}$ and $\delta_v$ are regular enough functions defined on $\Gamma_{v,0}$.
    The results on this work can be immediately extended to that case, which can be used as a resource in order to make the last condition on Theorem~\ref{thm:chambrion_controllability} hold.
  \end{remark}
}
    
  The unitaries above define closable sesquilinear forms associated with self-adjoint extensions of the Laplace-Beltrami operator.
  We refer to \cite[Section 5]{IbortLledoPerezPardo2015} for further details.
  Note that our use of \emph{boundary conditions} is different from the usual one: instead of using the boundary conditions to define the domain of a differential operator, by using the unitary operator $U$ we define a closed Hermitian sesquilinear form and consider the unique self-adjoint extension of the differential operator associated with it.
  However, as discussed at the beginning of this section, for functions in $\mathcal{H}^2 \cap \dom \Delta_U$ the usual sense of boundary conditions can be recovered and the boundary equation of Definition~\ref{def:boundary_equation_operator} is equivalent to the equations, cf.\ Corollary~\ref{corol:magnetic_laplacian_equivalence},
  \begin{equation}\label{eq:boundary_eqs}
    P \varphi = 0, \qquad P^\perp \dot{\varphi} = \mathcal{C}_U \varphi.
  \end{equation}
  The partial Cayley transform is
  \begin{equation} \label{eq:circuits_cayley}
    \mathcal{C}_U = 0_{L^2(\Gamma_{\text{ext}})} \oplus \bigoplus_{v \in V} \mathbb{T}_{v}^* \left[-\tan\left( \frac{\delta_{v}}{2} \right) P_{v}^\perp \right]
    \mathbb{T}_{v}.
  \end{equation}
  Because of this block structure, the equations \eqref{eq:boundary_eqs} hold \emph{block by block}, yielding the following conditions.
  For any vertex $v\in V$ choose one of the adjacent edges $e$ and define $\varphi|_v := e^{-i\chi_{v,e}}T_{v,e} \varphi_{v,e} \in H^{\sfrac{1}{2}}({\Gamma_{v,0}})$ and then
  \renewcommand\arraystretch{1.33}
  \begin{equation*}
    \begin{array}{c c}
      \displaystyle
      \varphi|_{\Gamma_{\text{ext}}} = 0, \\
      \displaystyle
      e^{-i\chi_{v,e'}}T_{v,e'} \varphi_{v,e'} = e^{-i\chi_{v,e}} T_{v,e}\varphi|_{v,e}, \quad e'\neq e\\
      \displaystyle
      \sum_{e' \in E_v} \frac{1}{|E_v|}e^{-i\chi_{v,e'}} T_{v,e'} \dot{\varphi}_{v,e'} = -\tan \left( \frac{\delta_v}{2} \right) \varphi|_v.
    \end{array}
  \end{equation*}
  \renewcommand\arraystretch{1}
  Analogously, for $\Phi \in \mathcal{H}^2 \cup \dom \Delta_{A,U}$ one gets the boundary conditions
  \renewcommand\arraystretch{1.33}
  \begin{equation*}
    \begin{array}{c c}
      \displaystyle
            \varphi|_{\Gamma_{\text{ext}}} = 0, \\
      \displaystyle
      e^{-i\chi_{v,e'}}T_{v,e'} \varphi_{v,e'} = e^{-i\chi_{v,e}} T_{v,e}\varphi|_{v,e}; \quad e\neq e'\\
      \displaystyle
      \sum_{e' \in E_v} \frac{1}{|E_v|}e^{-i\chi_{v,e'}} T_{v,e'} (\dot{\varphi}_A)_{v,e'} = -\tan \left( \frac{\delta_v}{2} \right) \varphi|_v.
    \end{array}
  \end{equation*}
  \renewcommand\arraystretch{1}
  where again $\dot{\varphi}_A = \gamma(\d_A\Phi(\nu))$.
  
  Once the relation imposed between functions on the junctions has been written explicitly, let
  us review the role of the parameters $\delta_v$ and $\chi_{v,e}$.
  It is clear that whenever
  $\delta_v = 0$, the boundary conditions are quasi-Kirchhoff boundary conditions.
  When, in addition to $\delta_v = 0$, we have $\chi_{v,e} = 0$, the boundary conditions are just Kirchhoff boundary conditions imposing continuity of the function $\Phi$ and the conservation of the so-called flux of the normal derivatives at that point.
  Regarding $\delta_v$, note that, when $\chi_{v,e} = 0$, the boundary conditions impose continuity for the function $\Phi$ on the connections but a net flux of the normal derivatives proportional to the trace of $\Phi$ at that vertex.
  In other words, $\delta_v$ represents a delta-like interaction supported on the connections.

  To illustrate these situations, let us now introduce some concrete systems to exemplify the construction above.
  In all the examples we will consider the canonically flat metric.
  \begin{figure}[t]
    \renewcommand{\thesubfigure}{\Alph{subfigure}}
    \begin{subfigure}[b]{0.45\textwidth}
      \centering
      \includegraphics[width=\textwidth]{img/example01_graph.pdf}
      \caption{}\label{fig:example_graph}
    \end{subfigure}
    \begin{subfigure}[b]{0.45\textwidth}
      \centering
      \includegraphics[width=\textwidth]{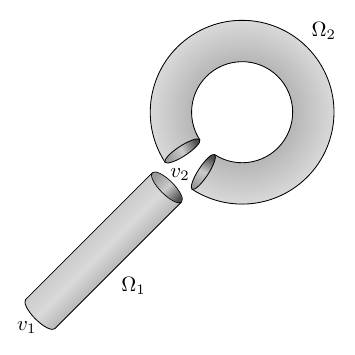}
      \caption{}\label{fig:example_cylinders}
    \end{subfigure}
    \caption{Examples of Thick Quantum Graphs defined on the same graph $G$.
    The Thick Quantum Graph (A) is one-dimensional (Quantum Graph).
    The Thick Quantum Graph (B) is two-dimensional.}
    \label{fig:examples}
  \end{figure}
  \begin{example}\label{ex:4}
    Consider the graph $G$ depicted in Fig~\ref{fig:example_graph}.
    For this first example we are going to consider the simplest class of Thick Quantum Graphs: those of dimension 1.
    For the manifold $\Omega$, consider two non-connected copies of the unit interval $\Omega_1 = [0,1]$, $\Omega_2 = [0,1]$.
    It is clear that $\partial\Omega_e = \{0_e, 1_e\}$, and the elements in the partition are the following:
    \begin{equation*}
      \Gamma_{v_1,e_1} = \{0_{e_1}\}, \quad
      \Gamma_{v_2,e_1} = \{1_{e_1}\}, \quad
      \Gamma_{v_2,e_2} = \{0_{e_2}, 1_{e_2}\}, \quad
      \Gamma_{\text{ext},e} = \emptyset \quad (e = e_1, e_2).
    \end{equation*}
    In such a case, our Thick Quantum Graph reduces to a Quantum Graph, and our definition of quasi-$\delta$, quasi-Kirchhoff and Kirchhoff boundary conditions respectively reduce to the equally named vertex conditions (cf.\ \cite{BerkolaikoKuchment2012}).
    In this simple case $L^2(\partial\Omega) = \mathbb{C}^4$, and there is no need for the unitaries $\mathbb{T}_v$, as all the pieces of the partition of the boundary are canonically diffeomorphic.
    At the vertex $v_1$, which is connected only to one edge, the boundary conditions become Robin boundary conditions.\triang
  \end{example}
  \begin{example}\label{ex:5}
    Consider the situation depicted in Fig.~\ref{fig:example_cylinders}.
    The graph associated with this Thick Quantum Graph is the same $G$ (depicted in Fig.~\ref{fig:example_graph}) from the previous example.
    The manifold is now given by two cylindrical surfaces, $\Omega_1$ and $\Omega_2$.
    Now $\partial\Omega_e = c_{e,-}\cup c_{e,+}$, where $c_{e,\pm}$ denote the circumferences at the edges of the cylindrical surfaces.
    As before, $\Gamma_{\text{ext},e} = \emptyset$ for both edges, and now
    \begin{equation*}
      \Gamma_{v_1,e_1} = c_{e_1,-}, \quad
      \Gamma_{v_2,e_1} = c_{e_1,+}, \quad
      \Gamma_{v_2,e_2} = c_{e_2,-}, \quad
      \Gamma_{v_2,\tilde{e}_2} = c_{e_2,+}.
    \end{equation*}
    If both cylinders have the same radius, then the isometries $T_{v,e}$ defining the unitaries $\mathbb{T}_v$ are just the natural ones transforming $c_{e,\pm}$ into the unit circle.
    If the radii are different, then the isometries include a scaling factor.\triang
  \end{example}
  \begin{example}\label{ex:6}
    Consider a Thick Quantum Graph with the same graph $G$ as in the previous examples.
    $\Omega$ is the disjoint union of two rectangles.
    The boundary of each rectangle is the disjoint union of four intervals.
    Two opposite sides of each rectangle will provide the elements $\Gamma_{v_1,e}$, $e = e_1,e_2$, of the partition while the other four sides, two for each rectangle, will be the elements $\Gamma_{\text{ext},e}$, $e = e_1,e_2,\tilde{e}_2$.
    {This is the Thick Quantum Graph depicted in Figure~\ref{fig:example-notation-bands}.}
    \triang
  \end{example}
  
%%%%%%%%%%%%%%%%%%%%%%%%%%%%%%%%%%%%%%%%%%%%%%%%%%%%%%%%%%%%%%%%%%%%%%%%%%%%%%%%%%%%%%%%%%%%%%%%%%%%%%%%%%%%%%%%%%%%%%%%%%%%%%%%%%%%%%%%%%%%%%%%%%%%%%%%%%%%%%%%%%%%%%%%%%%%%%%%%%%%%%%%%%%%%%%%%%%%%%%%%%%%%%%%%%%%%%%%%%%%%%%%%%%%%%%%%%%%%%%%%%%%%%%%%%%%%%%%%%%%%%%%

\section{Dynamics and Controllability on Thick Quantum Graphs} \label{sec:circuits}

{
\subsection{The controllability problem}
  Let us start this section by recalling some of the basic ideas on Quantum Control.
  Our concern here is to address the control in the infinite-dimensional quantum setting, that is, for quantum systems whose Hilbert spaces are infinite-dimensional.
  Therefore, we will not discuss the case of finite-dimensional quantum systems, to which the classic theory of control and all its results apply directly (we refer to \cite{DAlessandro2007} for a comprehensive and thorough revision of the topic).

  For the purposes of this work, we consider a system with a Hamiltonian $H(u)$ depending on a control parameter $u \in \mathcal{C}$ in a suitable space of controls $\mathcal{C} \subset \mathbb{R}$.
  The evolution from an initial state $\Psi_0$ is given by the Schrödinger equation
  \begin{equation}\label{eq:controlled-schrodinger}
    i \frac{\d}{\d t} \Psi(t) = H(u) \Psi(t), \qquad \Psi(0) = \Psi_0.
  \end{equation}
  The controllability problem is, for a given valid choice of initial and target states, to answer whether it is possible or not to construct a control curve $u(t)$ in $\mathcal{C}$ such that the solution of Eq.~\eqref{eq:controlled-schrodinger} for $H(u(t))$ starting from the initial state reaches the target state.

  Let us now present the precise definition of quantum control system we use on this work.
}

{
  \begin{definition}\label{def:control_system}
    Let $\mathcal{C}\subset\R^n$ be the space of controls.
    Let $\{\D_c\}_{c\in\mathcal{C}}$ be a family of dense subsets of $\H$ and let $\{H(c)\}_{c\in\mathcal{C}}$ be a family of self-adjoint operators such that, for each $c\in\mathcal{C}$, the operator $H(c)$ has domain $\D_c$.
    A \emph{quantum control system} is the family of dynamical systems defined by the Schrödinger equation with time-dependent Hamiltonians $H(u(t))$, where $u\colon [0,T] \subset \mathbb{R} \to \mathcal{C}$ are the control functions.
  \end{definition}

  As mentioned in the introduction, for quantum control systems defined on infinite-dimensional Hilbert spaces, the notion of exact controllability is too restrictive and needs to be relaxed.
  An appropriate notion that bypasses the no-go results and is suitable for the infinite-dimensional case is the concept of approximate controllability.
  {
  For a comprehensive discussion on this matter, see Sections~II and III in \cite{BoscainChambrionSigalotti2013}.
}

  \begin{definition}\label{strong_controllability}
    A quantum control system is \emph{approximately controllable} if for every $\Psi_0, \Psi_T \in \mathcal{H}$, with $\norm{\Psi_0}=\norm{\Psi_T}$, and every $\varepsilon >0$, there exists $T>0$ and a measurable control function $u:[0,T]\to\mathcal{C}$ such that the Schrödinger equation has a solution $\Psi(t)$ that satisfies
    \begin{equation*}
      \Psi(0) = \Psi_0 \quad
      \text{and} \quad
      \|\Psi_T - \Psi(T)\| < \varepsilon.
    \end{equation*}
  \end{definition}

  In \cite{ChambrionMasonSigalottiEtAl2009} Chambrion et al.\ study the approximate controllability of bilinear control systems defined as follows.
  \begin{definition}
  Let $r>0$ and assume that:
  \begin{enumerate}[%
    label=\textit{(A\arabic*)},
    nosep,
    labelindent=0.5\parindent,
    leftmargin=*
    ]
    \item\label{H1} $H_0, H_1$ are self-adjoint operators.
    \item\label{H2} There exists an orthonormal basis $\{\Phi_n\}_{n \in \mathbb{N}}$ of $\mathcal{H}$ made of eigenvectors of $H_0$.
    \item\label{H3} $\Phi_n \in \dom H_1$ for every $n \in \mathbb{N}$.
  \end{enumerate}
  A \emph{normal bilinear control system} is a quantum control system with controls given by \hbox{$\mathcal{C} = \{c\in\R | c<r\}$} and family of Hamiltonians $\{H(c)\}_{c\in \mathcal{C}}$ determined by $$H(c) = H_0 + c H_1.$$
  \end{definition}

  For bilinear control systems Chambrion et al.\ prove the following result.
  \begin{theorem}[Chambrion et al.\ {\cite[Thm.~2.4]{ChambrionMasonSigalottiEtAl2009}}] \label{thm:chambrion_controllability}
    Consider a normal bilinear control system as described above, and denote by $\lambda_n$ the eigenvalue of $H_0$ associated with the eigenfunction $\Phi_n$.
    Then, if the elements of the sequence $\{\lambda_{n+1} - \lambda_n\}_{n \in \mathbb{N}}$ are $\mathbb{Q}$-linearly independent and if $\langle \Phi_{n+1}, H_1 \Phi_n \rangle \neq 0$ for every $n \in \mathbb{N}$, the system is approximately controllable.
  \end{theorem}

  The conditions on the eigenvalues and eigenfunctions are not restrictive: as it is shown in \cite{BoscainCaponigroChambrionEtAl2012, MasonSigalotti2010, PrivatSigalotti2010,PrivatSigalotti2010b} they are satisfied generically and they can also be bypassed, cf.\ \cite[Section 6.1]{ChambrionMasonSigalottiEtAl2009}.
  {Moreover, taking non-constant values for the parameters defining quasi-$\delta$ boundary conditions can help to fulfil the last condition on certain circumstances, cf.\ Remark~\ref{remark:constant-delta-parameters}.}
  We will assume hereafter that these conditions are met.
}
\subsection{Dynamics}

  In the previous sections we have introduced Thick Quantum Graphs and families of self-adjoint differential operators on them that are well-defined quantum Hamiltonians.
  As stated in the introduction, the aim of this work is to study the feasibility of controlling a Thick Quantum Graph using the boundary conditions as controls; in particular, for showing the viability of this Quantum Control scheme, we are going to restrict ourselves to quasi-$\delta$ boundary conditions.
  We define now the two families of Quantum Control problems whose controllability is studied.

\begin{definition}\label{def:induction_control_system}
  Let $(G,\Omega,\Gamma)$ be a Thick Quantum Graph.
  Let $V$ be the vertex set of the graph $G$ and, for $v\in V$, let $E_v$ be the set of edges in $G$ that share the vertex $v$.
  Let $r$ be a positive real number, $A_0 \in \Lambda(\Omega)$ a smooth differential one-form and $\Theta_0 \in C^\infty (\Omega)$ a smooth function on $\Omega$ such that $\d\Theta_0 =A_0$.
  Let $U$ be the unitary operator defining quasi-$\delta$ boundary conditions with parameters $\chi_{v,e} =0$ and $\delta_v \in (-\pi, \pi)$, $v\in V$, $e\in E_v$.
A \emph{quantum induction control system} is a quantum control system with \hbox{$\mathcal{C} = \{ (a,b)\in\R^2\mid b <r\}$} and family of Hamiltonians $H(a,b) = \Delta_{aA_0,U} + b\Theta_0$, where $\Delta_{aA_0,U}$ is the self-adjoint extension of the magnetic Laplacian associated with the magnetic potential $aA_0\in\Lambda(\Omega)$ and the unitary $U$, and such that the control function is of the type
$$\begin{array}{rccc}u&:\R &\to& \mathcal{C}\\ & t &\mapsto & (a(t), b(t)),\end{array}$$
with $b(t) = \frac{\d a}{\d t}(t)$ almost everywhere.
\end{definition}

{
\begin{remark}
  Notice that the boundary conditions of the magnetic Laplacian, cf.\ Definition~\ref{def:boundary_equation_operator} and Theorem~\ref{corol:magnetic_laplacian_equivalence}, depend explicitly on the magnetic potential.
  Therefore, even if the parameters of the boundary conditions in Def.~\ref{def:induction_control_system} do not depend explicitly on time, the boundary conditions and thus the domains of the Hamiltonians do depend on time.
\end{remark}
}

The reason for letting $b(t)$ be the derivative of $a(t)$ almost everywhere is because we shall consider functions that are piecewise differentiable and thus the derivative might be undefined at certain points.

The term \emph{quantum induction} refers to the fact that the Schrödinger equation associated with quantum induction control systems is the one corresponding to a particle moving in the Thick Quantum Graph subject to the action of two objects: a time-dependent magnetic field concentrated on the loops of the graph, and an electric field whose strength is proportional to the variation of the magnetic field in the same way that is described by Faraday’s Induction Law.

\begin{definition}\label{def:boundary_control_system}
  Let $(G,\Omega,\Gamma)$ be a Thick Quantum Graph.
  Let $V$ be the vertex set of the graph $G$ and for $v\in V$ let $E_v$ be the set of edges in $G$ that share the vertex $v$.
  Let $\bar{\chi}_{v,e} \in [0,2\pi]$, $v\in V$, $e \in E_v$.
  A \emph{quasi-$\delta$ boundary control system} is the quantum control system with space of controls $\mathcal{C} = \R$ and family of Hamiltonians given by the Laplace-Beltrami operator on $\Omega$ with quasi-$\delta$ boundary conditions with parameters $\chi_{v,e} = c \bar{\chi}_{v,e} $ and $\delta_v \in (-\pi, \pi)$, $v\in V$, $e\in E_v$, and where $c\in \mathcal{C}$.
\end{definition}

  For addressing the controllability of quasi-$\delta$ boundary control systems, it is first necessary to study the existence of solutions for the corresponding Schrödinger equation.
  Quasi-$\delta$-type boundary control systems do not have constant form domain.
  The following theorem determines equivalent systems that will be proven to have constant form domain.

  \begin{theorem} \label{thm:quasi_delta_BCS_equivalence}
    Let $(G,\Omega,\Gamma)$ be a Thick Quantum Graph.
    Let $V$ be the vertex set of the graph $G$ and for $v\in V$ let $E_v$ be the set of edges in $G$ that share the vertex $v$.
    Let $U(t)$, $t\in\R$, be a time-dependent family of unitary operators defining quasi-$\delta$ boundary conditions with parameters $\chi_{v,e}(t)\in C_\mathrm{{pw}}^1(\mathbb{R})$ and $\delta_v \in (-\pi,\pi)$, $v\in V$ and $e\in E_v$.
    Let $H(t)$ be the time-dependent Hamiltonian defined by the family of Laplace-Beltrami operators $\Delta_{U(t)}$.
    Then there exists $\Theta(t) \in C_\mathrm{{pw}}^1(\mathbb{R}) \times C^\infty(\Omega)$ and a family of unitary operators $\{J(t)\}_{t\in\R}$ such that the curve $\Psi(t)$ is a {piecewise weak solution} of the Schrödinger equation determined by $H(t)$, cf.\ Definition~\ref{def:piecewise-sol}, if and only if $J(t)\Psi(t)$ is a {piecewise weak solution} of the Schrödinger equation determined by the time-dependent magnetic Hamiltonian
    \begin{equation}
      \tilde{H}(t) = \Delta_{A(t),\tilde{U}} + \Theta'(t),
    \end{equation}
    with the time-dependent magnetic potential $A(t) = \d\Theta(t)$, $\tilde{U}$ the unitary operator associated with \mbox{quasi-$\delta$} type boundary conditions with parameters $\tilde{\chi}_{v,e} = 0$, $\tilde{\delta}_v = \delta_v$, $v\in V$, $e\in E_v$, and $\Theta'|_I = \frac{\d}{\d t} \Theta|_{I}$ for any open interval $I$ on which every $\chi_{v,e}$ is $C^1$.
  \end{theorem}
  
  \begin{proof}
    The weak solutions of the Schrödinger equation determined by $H(t)$ satisfy
    \begin{equation} \label{eq:schrodinger_quasi_delta}
      \frac{\d}{\d t} \langle \xi, \Phi(t) \rangle = -i Q_{U(t)}(\xi, \Phi(t)),\qquad
      \text{for every }\xi \in \mathcal{H}^+,
    \end{equation}
    where $U(t)$ is the unitary defining quasi-$\delta$ boundary conditions with parameters $\chi_{v,e}(t), \delta_v$, cf.\ Definition~\ref{def:quasi_delta_BC}, and $Q_{U(t)}$ is the sesquilinear form associated with $\Delta_{U(t)}$ (see Def.~\ref{def:associated_form}).

    By definition, the closure of each $\Gamma_{v,e}\in \Gamma$, $e\in E$, $v\in V_e$ is contained in a connected open neighbourhood $\mathcal{V}_{v,e}$ and these are pairwise disjoint.
    Therefore, for each $e\in E$ there exists $\Theta_e\in C^\infty(\Omega_e)$ such that for any vertex $v \in V_e$ one has $T_{v,e}(\gamma(\Theta_e)|_{\Gamma_{v,e}}) = \chi_{v,e}$ where $T_{v,e}$ is the isometry defined in Eq.~\eqref{eq:definition_Te}.
    Moreover, if $\chi_{v,e}(\cdot)\in C_\mathrm{{pw}}^{a}(\R)$ for $a\in\mathbb{N}$, then $\Theta_e$ can be chosen such that $\Theta_e(\cdot) \in C_\mathrm{{pw}}^{a}(\R)\times C^\infty(\Omega_e)$.
    Let $\Theta(t) := \bigoplus_{e \in E} \Theta_e(t) \in C^\infty(\Omega)$, $t\in\R$.
    By construction $T_{v,e}(\gamma(\Theta(t))|_{\Gamma_{v,e}}) = \chi_{v,e}(t)$.
    Define the magnetic potential $A(t) = \d \Theta(t) = \bigoplus_{e \in E} \d \Theta_e(t)$ and denote by $J(t)$ the family of unitary transformations on $L^2(\Omega)$ defined by
    $J(t) = \bigoplus_{e \in E} J_e(t)$, where $J_e(t) : \Phi_e \in L^2(\Omega_e)
    \mapsto e^{-i\Theta_e(t)} \Phi_e \in L^2(\Omega_e)$.
    Some consequences follow straightforwardly:
    \begin{enumerate}[label=\textit{(\roman*)}]
      \item For every $t$ in an interval on which the functions $\chi_{v,e}(t)$ are $C^1$ and $\Phi \in H^1(\Omega)$, the product rule yields
        \begin{equation*}
          \d_{A(t)} (J(t) \Phi) = \d(J(t)\Phi) + i A(t) J(t)\Phi = J(t) \d \Phi.
        \end{equation*}
      \item For every $t \in I$ and $\Phi \in \dom Q_{U(t)}$, $\Psi = J(t) \Phi$ is in the constant form domain $\dom Q_{A(t), \tilde{U}}$.
        This is a direct application of Proposition~\ref{prop:magnetic_form_equivalence}.
      \item For a curve $\Phi(t) \in \dom Q_{U(t)}$ that is a weak solution of the Schrödinger equation
        \eqref{eq:schrodinger_quasi_delta} in an interval $I$, $\Psi(t) = J(t) \Phi(t)$ satisfies also in $I$
        \begin{equation*}
          \begin{alignedat}{2}
            \frac{\d}{\d t} \langle \xi, \Psi(t) \rangle
            &= \left\langle\xi, \left( \frac{\d}{\d t}J(t) \right) \Phi(t) \right\rangle + \left\langle \xi, J(t) \frac{\d}{\d t} \Phi(t)\right\rangle \\
            &= -i \langle \xi, \Theta'(t) \Psi(t) \rangle 
            -i \left\langle J(t)^{-1} \xi, \frac{\d}{\d t} \Phi(t) \right\rangle \\
            &= -i \left[ \langle \xi, \Theta'(t) \Psi(t) \rangle 
            + Q_{U(t)}(J(t)^{-1} \xi, J(t)^{-1} \Psi(t)) \right]
          \end{alignedat}
        \end{equation*}
        By Proposition~\ref{prop:magnetic_form_equivalence}, this is equivalent to
        \begin{equation*}
          \frac{\d}{\d t} \langle \xi, \Psi(t) \rangle
          = -i \left[
            Q_{A(t), \tilde{U}}(\xi, \Psi(t)) + \langle \xi, \Theta'(t) \Psi(t) \rangle 
          \right].
        \end{equation*}
        where $\tilde{U}$ is the unitary associated with $\delta$-type boundary conditions with parameters $\delta_v$.
    \end{enumerate}
    The result follows noticing that the sesquilinear form associated with $\tilde{H}(t)$ is
    \begin{equation*}
      h_t(\Psi,\Phi) = Q_{A(t), \tilde{U}}(\Psi, \Phi) + \langle \Psi, \Theta'(t) \Phi \rangle. \qedhere
    \end{equation*}
  \end{proof}

  This theorem allows us to convert the original quasi-$\delta$-type boundary control system into
  an induction control system.
  Part of the time-dependence we had in the domain of the initial problem has been transferred to the analytical form of the Hamiltonian, but we still have a system with time-dependent operator domain.
  However, the domains of the family of sesquilinear forms associated with the self-adjoint operators of the quantum induction control problem do not depend on the parameter $t$ and have therefore constant form domain.
  This is shown in the next result.

  \begin{corollary} \label{cor:magnetic_system_form_domain}
    Let $I \subset \mathbb{R}$ be a compact interval and let $H(t)$, $t\in I$, be the time-dependent Hamiltonian of the quantum induction control problem associated with a quasi-$\delta$ boundary control problem with control functions $\chi_{v,e}(t)$; $v\in V$, $e\in E_v$.
    Then $H(t)$ is a time-dependent Hamiltonian with constant form domain, cf.\ Def.~\ref{def:hamiltonian_const_form}.

    In addition, if $\chi_{v,e}(t) \in C_\mathrm{{pw}}^2(I)$, there exists a {piecewise weak solution} of the Schrödinger equation determined by $H(t)$.
    Moreover, if $\chi_{v,e}(t)$ is in $C^3(I)$, then the {piecewise weak solution} is also a strong solution of the Schrödinger equation.
  \end{corollary}
  \begin{proof}
    Notice that the unitary operators defining the boundary conditions for the quantum induction system do not depend on the parameters $\chi_{v,e}(t)$ and are therefore constant.
    Since the graph has a finite number of vertices, $\inf_{v\in V} |\delta_v \pm \pi| > 0$.
    From \cite[Prop.~3.11]{IbortLledoPerezPardo2015} it follows that the corresponding Cayley transforms satisfy trivially the inequality $\sup_t \|\mathcal{C}_{U}\| < \infty$.
    Since $\Omega$ is compact, $\Theta'(t)$ is a bounded operator on $L^2(\Omega)$.
    For each $t\in I$, $A(t)$ is a differentiable form on $\Omega$ and, since $I$ is a compact interval, $\sup_{t\in I} \norm{A}_{\infty} < \infty$.
    Corollary~\ref{corol:uniform_semibound_equivalence_formnorms} implies the uniform lower bound for $H(t)$.
    The form domain of $H(t)$ is
    \begin{equation*}
      \mathcal{H}^+ = \dom Q_{A(t),U} = \{\Phi \in H^1(\Omega) \mid P\varphi = 0\},
    \end{equation*}
    where $P$ is the orthogonal projector onto the eigenspace of $U$ associated with the eigenvalue -1 of the unitary operator (see Definition \ref{def:quasi_delta_BC}), which only depends on the topology of the graph $G$ associated with the Thick Quantum Graph.

    The existence of {piecewise weak solutions}, respectively strong solutions, of the Schrödinger equation is ensured by \cite[Theorem~5]{BalmasedaLonigroPerezPardo2022a} since the associated form is given by
    \begin{equation*}
      h_t(\Phi,\Phi) = \|\d\Phi\|^2 + \|\Phi A(t)\|^2 + 2\Im\langle \d\Phi, \Phi A(t) \rangle - \langle \varphi, \mathcal{C}_{U}\varphi \rangle + \langle \Phi, \Theta'(t)\Phi \rangle
    \end{equation*}
    with $\Theta(t)$ and $A(t)$ defined in Theorem~\ref{thm:quasi_delta_BCS_equivalence} and $\mathcal{C}_{U}$ given in Eq.~\eqref{eq:circuits_cayley}.
  \end{proof}

%%%%%%%%%%%%%%%%%%%%%%%%%%%%%%%%%%%%%%%%%%%%%%%%%%%%%%%%%%%%%%%%%%%%%%%%%%%%%%%%%%%%%%%%%%%%%%%%%%%%%%%%%%%%%%%%%%%%%%%%%%%%%%%%%%%%%%%%%%%%%%%%%%%%%%%%%%%%%%%%%%%%%%%%%%%%%%%%%%%%%%%%%%%%%%%%%%%%%%%%%%%%%%%%%%%%%%%%%%%%%%%%%%%%%%%%%%%%%%%%%%%%%%%%%%%%%%%%%%%%%%%%

\subsection{Controllability}

The approximate controllability problem for a quantum boundary control system consists in answering whether it is possible to drive the system from any initial state to a small neighbourhood of any target state by only modifying its boundary conditions.
In our particular setting of quasi-$\delta$-type boundary control systems, this is done by choosing a family of curves $\chi_{v,e}(t)$.

By Theorem~\ref{thm:quasi_delta_BCS_equivalence}, the control problem for a quasi-$\delta$-type boundary control system is closely related to the control problem for the associated induction control problem, given by the magnetic Hamiltonian
\begin{equation} \label{eq:magnetic_controlled_H}
  H(t) = \Delta_{A(t)} + \Theta'(t),
\end{equation}
with controls $A(t), \Theta'(t)$ such that $A(t) = \d\Theta(t)$ and with $\delta$-type boundary conditions.

We will prove now Theorem~\ref{thm:controllability_piecewise_smooth}, i.e., approximate controllability of quantum induction control systems.
In order to do that, we rely on Chambrion et al.'s theorem (i.e., Theorem~\ref{thm:chambrion_controllability}) and a stability Theorem proven in \cite{BalmasedaPhD}.
Notice that, even if the Hamiltonians of quantum induction control systems are similar to those of normal bilinear control system, the fact that the control function appears with its derivative does not allow for a direct application of Theorem~\ref{thm:chambrion_controllability}.
To circumvent this issue, we will proceed in two steps.
First, we will define an auxiliary system to which Chambrion et al.'s Theorem applies.
Then, we will use the controls provided by Theorem~\ref{thm:chambrion_controllability} to construct a sequence of Hamiltonians converging to the quantum induction one.
The next theorem collects some results and applies them to this sequence of Hamiltonians.
We shall use the scales of Hilbert spaces associated with the Hamiltonians with constant form domain.

 \begin{definition} \label{def:scales_hilbert}
    Let $I \subset \mathbb{R}$ a compact interval, let $\mathcal{H}^+$ be a dense subspace of $\mathcal{H}$ and ${H(t)}$, ${t \in I}$, a time-dependent Hamiltonian with constant form domain $\H^+$.
    The \emph{scale of Hilbert spaces} defined by $H(t)$, cf.\ \cite[Section 1.1]{Berezanskii1968}, is the triple of Hilbert spaces
    \begin{equation*}
      (\mathcal{H}^+, \langle \cdot, \cdot \rangle_{+,t}) \subset
      (\mathcal{H}, \langle \cdot, \cdot \rangle) \subset
      (\mathcal{H}^-_t, \langle \cdot, \cdot \rangle_{-,t}),
    \end{equation*}
    where $\langle \Psi, \Phi \rangle_{\pm,t} := \langle (H(t) + m + 1)^{\pm\sfrac{1}{2}}\Psi, (H(t) + m + 1)^{\pm\sfrac{1}{2}} \Phi \rangle$ and $\mathcal{H}^-_t$ denotes the closure of $\mathcal{H}$ with respect to the norm defined by $\|\Phi\|_{-,t}^2 \coloneqq \langle \Phi, \Phi \rangle_{-,t}$.
 We will denote by $(\argdot,\argdot)_t : \H^+_t\times\H^-_t \to \mathbb{C}$ the canonical pairings.
 \end{definition}

  An important property of the scales of Hilbert spaces, which will be needed, is the following one: $$\norm{\argdot}_{-,t}\leq \norm{\argdot}\leq \norm{\argdot}_{+,t}.$$

\begin{theorem} \label{lemma:magnetic_family_formlinear}
  Let $I \subset \mathbb{R}$ be a compact interval, let $A_0$ be a magnetic potential and let $\Theta_0$ be a function such that $\d\Theta_0 = A_0$.
  For $n \in \boldsymbol{N} \subset \mathbb{N}$, let $u_n(t), v_n(t) \in C_\mathrm{{pw}}^2(I)$ such that $\sup_{n,t} |u_n(t)| < \infty$ and $\sup_{n,t} |v_n(t)| < \infty$.
  For each $t \in I$, denote by $\Delta_{u_n(t) A_0}$ the magnetic Laplacian with potential $u_n(t)A_0$ and constant $\delta$-type boundary conditions.
  Define the Hamiltonians $H_n(t) \coloneqq \Delta_{u_n(t) A_0} + v_n(t)\Theta_0$.
  The following statements hold:
  \begin{enumerate}[label=\textit{(\alph*)},nosep,leftmargin=*,ref=\ref{lemma:magnetic_family_formlinear}\alph*]
    \item \label{lemma:magnetic_formlinear_dynamics}
      $\{H_n(t)\}_{n \in \boldsymbol{N}}$ is a family of form-linear, time-dependent Hamiltonians, and for each $n \in \boldsymbol{N}$ there exists a unitary propagator $U_n(t,s)$ that is a {piecewise weak solution} of Schrödinger equation for $H_n(t)$, cf.\ Definition~\ref{def:piecewise-sol}.
    \item \label{lemma:magnetic_formlinear_normequiv}
      {For any $n_0 \in \boldsymbol{N}$ and any $t_0 \in I$, there is a real number $c > 1$} independent of $n$ and $t$ such that
      \begin{equation*}
        c^{-1} \|\cdot\|_{\pm,n,t} \leq \|\cdot\|_{\pm,n_0,t_0} \leq c\|\cdot\|_{\pm,n,t}.
      \end{equation*}
      {
        We will take a reference $n_0 \in \boldsymbol{N}$ and $t_0 \in I$, and denote $\|\cdot\|_{\pm} \coloneqq \|\cdot\|_{\pm,n_0,t_0}$.
      }
    \item \label{lemma:magnetic_formlinear_stability}
      If, in addition, $\sup_{n \in \boldsymbol{N}} \sum_j \|u_n'\|_{L^1(I_j)} < \infty$ and $\sup_{n \in \boldsymbol{N}} \sum_j \|v_n'\|_{L^1(I_j)} < \infty$, where $\{I_j\}_{j=1}^\nu$ denotes the family of open intervals on which $u_n, v_n$ are differentiable, then there is a constant $L$ such that for every $n \in \boldsymbol{N}$ and $t, s \in I$ it holds
      \begin{alignat*}{2}
        \|U_n(t,s) - U_{n'}(t,s)\|_{+,-}
        \leq L &\left(\|u_n - u_{n'}\|_{L^1(s,t)}
        + \|u_n^2 - u_{n'}^2\|_{L^1(s,t)} \right. \\
        &\left.+ \|v_n - v_{n'}\|_{L^1(s,t)}\right),
      \end{alignat*}
      {where $\|\cdot\|_{+,-}$ denotes the norm on the space of bounded operators from $\mathcal{H}^+$ to $\mathcal{H}^-$.}

  \end{enumerate}
\end{theorem}
\begin{proof}
  The sesquilinear form defined by $H_n$ can be written as
  \begin{equation*}
    h_{n,t}(\Phi,\Phi) = h_0(\Phi,\Phi) + u_n(t)^2 \|\Phi A_0\|^2 + 2u_n(t) \Im\langle \d\Phi, \Phi A_0 \rangle + v_n(t) \langle \Phi, \Theta_0 \Phi \rangle,
  \end{equation*}
  where $h_0(\Phi,\Phi) = \|\d\Phi\|^2 - \langle \varphi, \mathcal{C}_{U} \varphi \rangle$ does not depend on $t$ since $\delta_v$ is constant (see Eq.~\eqref{eq:circuits_cayley}).
  Since $\Omega$ is compact, $\Theta_0$ and $A_0$ are bounded and the boundedness of $u_n(t), v_n(t)$ and Corollary~\ref{corol:uniform_semibound_equivalence_formnorms} shows that the Hamiltonians $H_n(t)$ are semibounded from below uniformly.
  By Corollary~\ref{cor:magnetic_system_form_domain}, \emph{(a)} follows.

  Property \emph{(b)} follows from applying \cite[Proposition~4.4.3]{BalmasedaPhD}.
  Now, since $\Omega$ is bounded, $h_1(\Phi,\Phi) \coloneqq \|\Phi A_0\|^2$ and $h_3(\Phi,\Phi)\coloneqq \langle \Phi, \Theta_0 \Phi \rangle$ are bounded with respect to the norm $\|\cdot\|\leq \|\cdot\|_+$.
  Finally, since $\|\cdot\|_{+} \sim \|\cdot\|_1$ (cf.\ Theorem~\ref{thm:quadratic_form_H1bounded}), the form $h_2(\Phi,\Phi) \coloneqq 2 \Im \,\langle \d\Phi, \Phi A_0 \rangle$ is bounded with respect to $\|\cdot\|_+$.
  Therefore, there is $K$ such that $h_i(\Phi,\Phi) \leq K \|\Phi\|_+$ for $i = 1,2,3$.
  Since $\sup_{n,t} |u_n(t)| <\infty$, $I$ being compact, and $\sup_{n} \sum_j \|u_n'(t)\|_{L^1(I_{j})} <\infty$, it follows that $\sup_{n} \sum_j \|\frac{\d}{\d t}[u_n]^2\|_{L^1(I_j)} <\infty$.
  Therefore, \cite[Theorem~4.4.5]{BalmasedaPhD} applies, which concludes the proof.
\end{proof}

\begin{proof}[Proof of Theorem~\ref{thm:controllability_piecewise_smooth}]
  For any $u_0 > 0$, define the auxiliary system with Hamiltonian
    \begin{equation*}
      H_0(t) = \Delta_{u_0 A_0} + v(t)\Theta_0.
    \end{equation*}
    for some magnetic potential $A_0$ and $\Theta_0$ such that $d\Theta_0=A_0$.
    We have omitted the subindex $U$ denoting the boundary conditions of the magnetic Laplacian as it will remain fixed.
    Since the Thick Quantum Graph is defined on a compact manifold, $\Theta_0$ defines a bounded potential.
    Moreover, since $u_0A_0$ is fixed, the operator domain of $H_0(t)$ does not depend on $t$, and $\Delta_{u_0 A_0}$ has compact resolvent since the Thick Quantum Graph is a compact manifold.
    We will assume that the conditions on the eigenvalues and eigenfunctions of Theorem~\ref{thm:chambrion_controllability} are met.
    These conditions are met generically in the systems under study.

    Hence, for every initial and target states $\Psi$, $\Psi_T$ with $\|\Psi_T\|=\|\Psi\|$, every $\varepsilon > 0$ and every $r > 0$, there exists $T>0$ and $v(t): [0, T] \to (0, r)$ piecewise constant such that the evolution induced by $H_0(t)$, $\Psi_0(t) = U_0(t,0)\Psi$, satisfies $\Psi_0(0) = \Psi$ and
    \begin{equation*}
      \left\| \Psi_0(T) - \Psi_T \right\| < \frac{\varepsilon}{2}.
    \end{equation*}
    Note that, since the operator domain is fixed, for a piecewise constant $v(t)$ there exist a {piecewise weak solution} of the Schrödinger equation determined by $H_0(t)$.
    
    Now we will construct a sequence of Hamiltonians whose dynamics will converge to the auxiliary one.
    For each $n \in \mathbb{N}$, divide the time interval $I=[0,T]$ into $n$ pieces of length $\tau = T/n$.
    Let $\{I_{n,j}\}_{j=1}^{\nu_n}$ be the coarsest refinement of the partition $\{((k-1)\tau, k\tau)\}_{k=1}^{n}$ such that $v(t)$ is constant with value $v_{n,j}$ on $I_{n,j}$ for $j = 1, 2, \dots, n$.
    Define $t_{n,j}$ such that $I_{n,j} = (t_{n,j}, t_{n,j+1})$.
    By construction $t_{n, j+1} - t_{n,j} \leq \tau$ for $1\leq j \leq \nu_n$.

    For $1 \leq j \leq \nu_n$, define the functions $u_{n,j}: [0,T] \to \mathbb{R}$ by
    \begin{equation*}
      u_{n,j}(t) =
      \begin{cases}
        0 & \text{if }t \notin I_{n,j}, \\
        u_0 + \int_{t_{n,j}}^t v(s) \,\d s &\text{if } t \in I_{n,j}.
      \end{cases}
    \end{equation*}
    Take $u_n(t) = \sum_{j=1}^{\nu_n} u_{n,j}(t)$.
    By Theorem~\ref{lemma:magnetic_formlinear_dynamics}, there is a unitary propagator
    $U_n(t,s)$, $t,s \in I$, that is a {piecewise weak solution} of the Schrödinger equation with Hamiltonian
    \begin{equation}\label{eq:induction_hamiltonian}
      H_n(t) = \Delta_{u_n(t)A_0} + v(t)\Theta_0.
    \end{equation}
    Note that, by definition, $u_n'(t) = v(t)$ for almost every $t \in I$ and for every $n$ and every $t \in I$ we have $|u_n'(t)| = |v(t)| < r$, $|u_n(t)| \leq u_0 + r T$ and $v'(t) = 0$.

    For every $t \in \bigcup_{j=1}^{\nu_n} I_{n,j}$ we have 
    \begin{equation*}
      |u_n(t) - u_0| = \left|\int_{t_{n,j}}^{t} v(s)\,\d s \right|
      \leq \frac{rT}{n},
    \end{equation*}
    and
    \begin{equation*}
      |u_n^2(t) - u_0^2| = 2u_0\int_{t_{n,j}}^t v(s)\,\d s + \left( \int_{t_{n,j}}^t v(s)\,\d s \right)^2
      \leq u_0 \frac{rT}{n} + \frac{r^2T^2}{n^2}.
    \end{equation*}
    It follows
    \begin{equation} \label{eq:convergence_un}
      \|u_n - u_0\|_{L^1(I)} \leq \frac{rT^2}{n} \quad\text{and}\quad
      \|u_n^2 - u_0^2\|_{L^1(I)} \leq u_0 \frac{rT^2}{n} + \frac{r^2T^3}{n^2}.
    \end{equation}
    Finally, Theorem~\ref{lemma:magnetic_formlinear_stability} applies, yielding
    \begin{equation*}
      \|U_0(t, s) - U_n(t, s)\|_{+,-} < L (\|u_n(t) - u_0\|_{L^1(s,t)} + \|u_n(t)^2 - u_0^2\|_{L^1(s,t)}),
    \end{equation*}
    where the constant $L$ is independent of $t, s, n$.
    By Equation~\eqref{eq:convergence_un} it follows
    \begin{equation*}
      \lim_{n \to \infty} \|U_0(t, s) - U_n(t, s)\|_{+,-} = 0
    \end{equation*}
    uniformly on $s,t \in I$.
    Suppose that $\Psi_0, \Psi_T\in \H^+$.
    We get that, for $n$ large enough,
  \begin{equation*}
  \norm{U_n(T,0)\Psi_0-\Psi_T}_- \leq \norm{U_n(T,0)\Psi_0 - U_0(T,0)\Psi_0}_- + \norm{U_0(T,0)\Psi_0-\Psi_T}_- < \varepsilon.
  \end{equation*}

Finally, we have to show that the approximation can be done with respect to the norm $\norm{\cdot}$ of the Hilbert space $\H$.
Remember that $\norm{\Psi_T} = \norm{\Psi_0}$.
By the results above, for any $\varepsilon>0$ there exists a piecewise linear function $u:I\mapsto\R$ satisfying the conditions of the Theorem such that there is a {piecewise weak solution} $U_u(t,s)$, $t,s \in I$, of the Schrödinger equation that satisfies
$$\norm{\Psi_T - U_u(T,0)\Psi_0}_-<\frac{\varepsilon^2}{2\norm{\Psi_T}_+}.$$

We have

\begin{align*}
  \norm{\Psi_T - U_u(T,0)\Psi_0}^2 &= \scalar{\Psi_T}{\Psi_T-U_u(T,0)\Psi_0} + \norm{\Psi_0}^2 + \scalar{-U_u(T,0)\Psi_0}{\Psi_T}\\
      &= \scalar{\Psi_T}{\Psi_T-U_u(T,0)\Psi_0} + \norm{\Psi_0}^2 + \scalar{\Psi_T-U_u(T,0)\Psi_0}{\Psi_T} - \norm{\Psi_T}^2\\
      &= 2\Re \scalar{\Psi_T}{\Psi_T-U_u(T,0)\Psi_0}\\
      &\leq 2\norm{\Psi_T}_+\norm{\Psi_T - U_u(T,0)\Psi_0}_- < \varepsilon^2.
\end{align*}
Since $\H^+$ is dense in $\H$, an $\varepsilon/2$ argument shows that the approximation holds for any \mbox{$\Psi_0, \Psi_T \in \H$}.
\end{proof}

\begin{proposition}\label{prop:induction_app_controllability_H-}
  Let $r \in \mathbb{R}$ be a positive real number, $u_1,u_0\in\R$ and $\mathcal{C} = \{(a,b) \in \mathbb{R}^2 \mid b<r\}$.
  A quantum induction control system is approximately controllable with control function
  \begin{equation*}
    \begin{array}{rccc}
      u&:[0,T] &\to& \mathcal{C}\\
       & t &\mapsto & (a(t), b(t)),
    \end{array}
  \end{equation*}
  piecewise linear and such that $b(t) = \frac{\d a}{\d t}(t)$ almost everywhere, $a(T)=u_1$ and $a(0)=u_0$.
\end{proposition}

\begin{proof}
  Let $H(a,b)$ be the family of Hamiltonians of the quantum induction control system, and define $\tilde{H}(a) := H(a,0)$ for $a\in\R$.
  A direct application of Stone's Theorem shows that, for any $\Phi\in\H$,
  $$\lim_{p\to0} \norm{(e^{-i\tilde{H}(a)p} - \mathbb{I})\Phi} = 0.$$
  Now take $\Psi_0,\Psi_T\in\H$ with $\norm{\Psi_T} = \norm{\Psi_0}$ and $q>0$.
  By Theorem~\ref{thm:controllability_piecewise_smooth}, for any $\varepsilon>0$ there exist $\tilde{T}>0$ and a piecewise linear control function $u\colon[q,\tilde{T}]\to\R$, with $|\frac{du}{dt}|\leq r$ almost everywhere such that the {piecewise weak solution} $\tilde{U}(\tilde{T},q)$ of the quantum induction control problem satisfies
  $$\norm{\Psi_T - \tilde{U}(\tilde{T},q)e^{-i\tilde{H}(u_0)q}\Psi_0} < \frac{\varepsilon}{3}.$$
  Define $\xi\in\H$ by $\xi:= \Psi_T - \tilde{U}(\tilde{T},q)e^{-i\tilde{H}(u_1)q}\Psi_0$.
  For any $p>0$ we have
  \begin{equation*}
  \norm{\Psi_T - e^{-i\tilde{H}(u_1)p}\tilde{U}(\tilde{T},q)e^{-i\tilde{H}(u_0)q}\Psi_0} \leq \norm{(\mathbb{I} - e^{-i\tilde{H}(u_1)p} )\Psi_T} + \norm{(e^{-i\tilde{H}(u_1)p} - \mathbb{I})\xi} + \norm{\xi}.
  \end{equation*}
  Taking $p$ small enough the right hand side can be made smaller than $\varepsilon$.
  Defining $T:=\tilde{T}+p$ the statement follows.
\end{proof}

\begin{proof}[Proof of Theorem~\ref{thm:controllability_boundary}]
  Let $V$ be the vertex set of the Thick Quantum Graph and, for $v\in V$, let $E_v$ be the set of edges of the Thick Quantum Graph that share the vertex $v$.
  Let $\delta_v\in (-\pi,\pi)$ for $v\in V$.
  Let $U_0$ be the admissible unitary operator defining quasi-$\delta$ boundary conditions with parameters $\chi_{v,e} = u_0\bar{\chi}_{v,e}$ and $\delta_v$, $v\in V$ and $e\in E_v$; $U_1$ is defined analogously with $u_0$ replaced by $u_1$.
  Let $\Psi_i\in\dom(\Delta_{U_i})$, $\norm{\Psi_i}=1$, $i=0,1$.

  By Proposition~\ref{prop:magnetic_form_equivalence}, the sesquilinear forms associated with the Laplace-Beltrami operators $\Delta_{U_i}$, $i=0,1$, are unitarily equivalent to the ones associated with magnetic Laplacians with magnetic potential $u_iA_0$ and quasi-$\delta$ boundary conditions determined by $\chi_{v,e}=0$ and $\delta_v$, $v\in V$ and $e\in E_v$.
  The unitary operators that implement the equivalence are completely determined by the values of the parameters $u_0$ and $u_1$, and will be respectively denoted by $J_{u_0}$ and $J_{u_1}$.

  Take $\Phi_i=J_{u_i}\Psi_i\in \dom (\Delta_{u_i A_0})$, $i=0,1$.
  By Theorem~\ref{thm:quasi_delta_BCS_equivalence} the solutions of the Schrödinger equation of the quasi-$\delta$ boundary control system are isomorphic to the solutions of a quantum induction control system, the isomorphism being a time-dependent unitary operator $J(t)$.
  Applying Theorem~\ref{thm:controllability_piecewise_smooth} to the initial and target states $\Phi_0$ and $\Phi_1$ it follows that, for any $\varepsilon>0$, there exists $T>0$ and a piecewise linear control function $u:[0,T]\to\R$ with $|\frac{du}{dt}| < r$, $u(T)=u_1$ and $u(0)=u_0$ such that 
  $$\|\Phi_1 - U(T,0)\Phi_0\| <\varepsilon,$$
  where $U$ is the unitary propagator that solves the Schrödinger equation of the quantum induction system.
  The curve $J^\dagger(t)U(t,0)\Phi_0$ is a solution of the Schrödinger equation of the quasi-$\delta$ boundary control system.
  For each $t\in[0,T]$, the unitary operator $J(t)$ depends only on the value of the magnetic potential at time $t$ and in fact $J(T)=J_{u_1}$ and $J(0) = J_{u_0}$.
  Therefore, we have
  $$J^\dagger(0)U(0,0)\Phi_0 = J^\dagger_{u_0}\Phi_0= \Psi_0$$
  and using the fact that $J_{u_1}$ is a unitary operator on $\H$
  \begin{equation*}
    \|\Psi_1 - J^\dagger(T)U(T,0)\Phi_0\| \leq \|\Phi_1 - U(T,0)\Phi_0\| < \varepsilon
  \end{equation*}
  as we wanted to show.
\end{proof}

\section{Conclusions and some further applications to approximate controllability}\label{sec:examples}
  In the previous sections we have obtained fairly general results that can be applied to a wide family of systems.
  However, the controllability results developed hold only with {piecewise weak solutions} of the Schrödinger equation, cf.\ Definition~\ref{def:piecewise-sol}, since the time-dependence of the Hamiltonians' domains hinders the existence of solutions in the strong sense.
  The main obstruction is the fact that the controls obtained by applying Theorem~\ref{thm:chambrion_controllability} are piecewise constant.
  If they were smooth, the results in \cite{BalmasedaLonigroPerezPardo2022a, Kisynski1964} would ensure that the {piecewise weak solutions} are also strong solutions.
  A first attempt to get approximate controllability with strong solutions of the Schrödinger equation would be to use the stability results developed in \cite{BalmasedaPhD} to extend Theorem~\ref{thm:controllability_piecewise_smooth} to the case of smooth controls.
  A straightforward application of this procedure is not possible for these systems.
  Piecewise constant functions do not have $L^1$ integrable weak derivatives, and the fact that the control and its derivative appear on the Hamiltonian of quantum induction systems forbids the application of these results.
  Nevertheless, the techniques developed can be applied to other related systems for which it is possible to obtain approximate controllability with strong solutions.
  Of course, the following results also apply to the particular case in which the Thick Quantum Graph is a Quantum Graph like in the Examples of Section~\ref{sec:preliminaries}.

  \begin{definition} \label{def:electromagnetic_control_system}
    Let $\mathcal{G}$ be a Thick Quantum Graph, let $A_0$ be a magnetic potential on $\mathcal{G}$, and let $\Theta_0$ be such that $A_0 = d\Theta_0$.
    We call \emph{electromagnetic quantum control system} to the quantum control system with real control function $u(t)$ and Hamiltonian
    \begin{equation*}
      H(t) = \Delta_{A_0} + u(t) \Theta_0,
    \end{equation*}
    where $\Delta_{A_0}$ is the magnetic Laplacian operator on $\mathcal{G}$ with quasi-$\delta$ boundary conditions with fixed parameters $\chi_{v,e}$, $\delta_v$.
  \end{definition}

  Using the ideas from the previous sections, this system can be proven to be approximately controllable with smooth control functions.

  \begin{proposition}\label{prop:6.2}
    Let $c > 0$ be fixed, and $u_0, u_1 \in (0, c)$.
    The electromagnetic quantum control system is approximately controllable with smooth control function $u: [0, T] \to (0, c)$ such that $u(0) = u_0$ and $u(T) = u_1$.
  \end{proposition}
  \begin{proof}
    Theorem~\ref{thm:chambrion_controllability} ensures approximate controllability with piecewise constant controls $\tilde{u}: [0, T] \to (0, c)$, and an argument analogous to Proposition~\ref{prop:induction_app_controllability_H-} shows that the controls can be chosen such that $\tilde{u}(0) = u_0$ and $\tilde{u}(T) = u_1$.
    
    Since piecewise constant functions can be approximated on $L^1$ by smooth functions, the result follows by a straightforward application of Theorem~\ref{lemma:magnetic_family_formlinear}.
  \end{proof}

  A relation analogous to the one between quantum induction control systems and boundary control systems can be established for electromagnetic quantum control systems.

  \begin{proposition}\label{prop:6.3}
    Consider an electromagnetic quantum control system with magnetic potential $A_0 = d\Theta_0$, and denote by $H(t)$ its Hamiltonian.
    Let $J(t)$ be the family of unitary operators defined edgewise by $J(t): \Phi \in \mathcal{H} \mapsto e^{iu(t)\Theta_0} \Phi \in \mathcal{H}$.
    A curve $\Phi(t)$ is a solution of the Schrödinger equation with Hamiltonian $H(t)$ if and only if $\Psi(t) \coloneqq J(t)\Phi(t)$ is a solution of the Schrödinger equation with Hamiltonian
    \begin{equation*}
      \tilde{H}(t) = \Delta_{A(t)},
    \end{equation*}
    where $\Delta_{A(t)}$ is the magnetic Laplacian with magnetic field $A(t) = (1 + u(t)) A_0$ and $t$-dependent quasi-$\delta$ boundary conditions with parameters $(1 + u(t))\chi_{v,e}$, $\delta_v$.
  \end{proposition}
  \begin{proof}
    This result is a consequence of Proposition~\ref{prop:magnetic_form_equivalence}.
    Using the relation between the operator domain and the form domain, it can be shown that $\Phi \in \dom H(t)$ if and only if $J(t)\Phi \in \dom \tilde{H}(t)$.
    Moreover, by the chain rule, it follows
    \begin{equation*}
      \frac{\d}{\d t} \Psi(t) = -i J(t) \Delta_{A_0} J(t)^{-1} \Psi(t) = -i \Delta_{(1+u(t))A_0} \Psi(t).
      \qedhere
    \end{equation*}
  \end{proof}

  This relationship motivates the following definition.
  \begin{definition}
    Let $\mathcal{G}$ be a Thick Quantum Graph, let $A_0$ be a magnetic potential on $\mathcal{G}$, and let $\Theta_0$ be such that $A_0 = d\Theta_0$.
    We call \emph{magnetic boundary control system} to the quantum control system with real control function $u(t)$ and Hamiltonian
    \begin{equation*}
      \tilde{H}(t) = \Delta_{(1 + u(t))A_0},
    \end{equation*}
    where $\Delta_{A_0}$ is the magnetic Laplacian operator on $\mathcal{G}$ with quasi-$\delta$ boundary conditions with parameters $(1 + u(t))\chi_{v,e}$, $\delta_v$.
  \end{definition}

  Combining Proposition~\ref{prop:6.2} and Proposition~\ref{prop:6.3} leads to the following approximate controllability result on magnetic boundary control systems:

  \begin{proposition}\label{prop:6.5}
    Let $c > 0$ be fixed, and $u_0, u_1 \in (0, c)$.
    The magnetic boundary control system is approximately controllable with smooth control function $u: [0, T] \to (0, c)$ such that $u(0) = u_0$ and $u(T) = u_1$.
  \end{proposition}

%%%%%%%%%%%%%%%%%%%%%%%%%%%%%%%%%%%%%%%%%%%%%%%%%%%%%%%%%%%%%%%%%%%%%%%%%%%%%%%%%%%%%%%%%%%%%%%%%%%%%%%%%%%%%%%%%%%%%%%%%%%%%%%%%%%%%%%%%%%%%%%%%%%%%%%%%%%%%%%%%%%%%%%%%%%%%%%%%%%%%%%%%%%%%%%%%%%%%%%%%%%%%%%%%%%%%%%%%%%%%%%%%%%%%%%%%%%%%%%%%%%%%%%%%%%%%%%%%%%%%%%%

\subsection*{Acknowledgements}

  A.B.\ and J.M.P.P.\ acknowledge support provided by the ``Agencia Estatal de Investigación (AEI)'' Research Project PID2020-117477GB-I00, by the QUITEMAD Project P2018/TCS-4342 funded by the Madrid Government (Comunidad de Madrid-Spain) and by the Madrid Government (Comunidad de Madrid-Spain) under the Multiannual Agreement with UC3M in the line of ``Research Funds for Beatriz Galindo Fellowships'' (C\&QIG-BG-CM-UC3M), and in the context of the V PRICIT (Regional Programme of Research and Technological Innovation).
  J.M.P.P acknowledges financial support from the Spanish Ministry of Science and Innovation, through the ``Severo Ochoa Programme for Centres of Excellence in R\&D'' (CEX2019-000904-S).
  A.B.\ acknowledges financial support from the Spanish Ministry of Universities through the UC3M Margarita Salas 2021-2023 program (``Convocatoria de la Universidad Carlos III de Madrid de Ayudas para la recualificación del sistema universitario español para 2021-2023''), and from ``Universidad Carlos III de Madrid'' through Ph.D.\ program grant PIPF UC3M 01-1819, UC3M mobility grant in 2020 and from the EXPRO grant No.\ 20-17749X of the Czech Science Foundation.
  D.L.\ was partially supported by ``Istituto Nazionale di Fisica Nucleare'' (INFN) through the project ``QUANTUM'' and the Italian National Group of Mathematical Physics (GNFM-INdAM), and acknowledges support by MIUR via PRIN 2017 (Progetto di Ricerca di Interesse Nazionale), project QUSHIP (2017SRNBRK), and by European Union–NextGenerationEU (CN00000013 – ``National Centre for HPC, Big Data and Quantum Computing'').
  He also thanks the Department of Mathematics at ``Universidad Carlos III de Madrid'' for its hospitality.

\printbibliography

\end{document}